%% file: main.tex
\pdfoutput=1
\RequirePackage[svgnames,dvipsnames]{xcolor}
\documentclass[journal,onecolumn,12pt,draftclsnofoot]{ieeecolor}
\usepackage{lcsys}
\usepackage{textcomp}
\usepackage{amssymb, amsfonts, mathtools, xargs, tensor, cite, stmaryrd, mathrsfs, algpseudocode, algorithm, graphicx}
\usepackage[unicode=true, colorlinks=false, hidelinks]{hyperref}

\newcommand{\R}{\mathbb{R}}

\newcommand{\tr}{\operatorname{tr}}

\newcommand{\bmax}{\beta_{\max}}

\newtheorem{assumption}{Assumption}
\newtheorem{definition}{Definition}
\newtheorem{remark}{Remark}

\newtheorem{lemma}{Lemma}

\newtheorem{theorem}{Theorem}

\newtheorem{example}[theorem]{Example}

\usepackage{tikz}
\usepackage{pgfplots}
\usepgfplotslibrary{groupplots}
\usepackage{pgfplotstable}
\pgfplotsset{compat=1.18}
\usetikzlibrary{spy}
 \usepackage{booktabs}
 \usepackage{multirow}
 \usepackage{siunitx}
 
\begin{document}

\title{On sparsity and directional forgetting in adaptive control}
\author{Tochukwu E. Ogri$^{1}$, Trivikram Satharasi$^{1}$, Muzaffar Qureshi$^{1}$, Kyle Volle$^{2}$, Rushikesh Kamalapurkar$^{1}$
\thanks{This research was supported in part by the Air Force Research Laboratory under contract number FA8651-24-1-0019. Any opinions, findings, or recommendations in this article are those of the author(s), and do not necessarily reflect the views of the sponsoring agency.}% 
\thanks{$^{1}$ Department of Mechanical and Aerospace Engineering, University of Florida, email: {\tt\footnotesize \{tochukwu.ogri, t.satharasi, muzaffar.qureshi, rkamalapurkar\} @ufl.edu}.}%
\thanks{$^{2}$ Torch Technologies, Shalimar, FL, USA, email: {\tt \footnotesize Kyle.Volle@torchtechnologies.com.}}}
\maketitle
\thispagestyle{empty}
\pagestyle{empty}

\begin{abstract}
This paper develops a sparsity-promoting memory regressor extension (MRE) adaptation law with directional forgetting for nonlinear control-affine systems with linearly parameterized uncertainty. The objective is to use directional forgetting to selectively discount obsolete information and leverage $\ell_1$ regularization to promote sparsity of the parameter estimates. While  $\ell_1$ regularization has been applied to the system identification problem in an offline setting, a contribution of this paper is to develop a recursive least squares update law to implement  $\ell_1$ regularization in online adaptive control. In particular, we show that $\ell_1$-regularized recursive least squares is realized via a sliding mode update law. A nonsmooth Lyapunov-based stability analysis is then used to show that the tracking and parameter estimation errors are ultimately bounded under a subspace excitation condition. Simulation results on a Van der Pol oscillator demonstrate the ability of the developed sparsity-promoting MRE controller to recover sparse dynamics while maintaining stable tracking.
\end{abstract}

\section{Introduction}

%In many dynamical systems, accurate mathematical models are unavailable due to unknown, noisy, or unmeasurable parameters. Controllers designed using such imperfect models may therefore fail to achieve the desired performance objectives. Adaptive control addresses this limitation by estimating unknown system parameters online within a Lyapunov stability framework, enabling regulation and trajectory tracking despite modeling uncertainty \cite{SCC.Anderson1977, SCC.Green.Moore1986, SCC.Astrom.Wittenmark2008, SCC.Krstic.Kanellakopoulos.ea1995, SCC.Ioannou.Sun1996}. 
Classical adaptive control methods, including gradient adaptation, least-squares estimation, model reference adaptive control (MRAC), and self-tuning regulators (STR), provide rigorous stability guarantees and have been successfully applied to a wide range of uncertain systems. However, a fundamental limitation of these approaches is that convergence of the parameter estimates generally requires the regressor to satisfy the persistent excitation (PE) condition \cite{SCC.Astrom.Wittenmark2008, SCC.Krstic.Kanellakopoulos.ea1995, SCC.Ioannou.Sun1996, SCC.Kamalapurkar2017, SCC.Pan.Shi.ea2024}, which is difficult to verify in practice, particularly for nonlinear systems.

To reduce reliance on PE, considerable effort has been devoted to improving the identifiability properties of adaptive learning schemes. Dynamic regressor extension and mixing (DREM), for example, transforms the original vector regression problem into a set of decoupled scalar regression problems, enabling parameter convergence under weaker excitation conditions \cite{SCC.Pyrkin.Bobtsov.ea2019}. Memory regressor extension (MRE) supplements current measurements with historical behavior of the system, enabling parameter convergence under relaxed excitation conditions such as interval excitation \cite{SCC.Katiyar.BasuRoy.ea2022, SCC.Pan.Shi.ea2024} (also known as finite excitation (FE) \cite{SCC.Parikh.Kamalapurkar.ea2019,SCC.Ogri.Bell.ea2023,SCC.Qureshi.Ogri.ea2025,SCC.Xing.Na.ea2025}). 
%Another commonly used MRE approach is concurrent learning (CL), which incorporates recorded trajectory data directly into the adaptive update law. By storing informative data over a finite time interval and constructing an auxiliary regression condition, CL enables parameter convergence under a finite excitation condition \cite{SCC.Parikh.Kamalapurkar.ea2019,SCC.Ogri.Bell.ea2023,SCC.Qureshi.Ogri.ea2025}. 

Although MRE approaches enable parameter convergence under relaxed excitation conditions, they can be unresponsive to changes in system parameters since stored data becomes inconsistent with the current system dynamics when the parameters change. This phenomenon is often referred to as loss of alertness \cite{SCC.Lee.Shin.ea2019}. To solve this problem, forgetting mechanisms have been proposed, including uniform forgetting strategies with variable forgetting factors \cite{SCC.Pan.Aranovskiy.ea2019} and directional forgetting methods \cite{SCC.Cao.Schwartz.ea2000, SCC.Bittanti.Bolzern.ea1990, SCC.Lee.Shin.ea2019, SCC.Zhu.Yu.ea2022, SCC.Wen.Xing.ea2025} that selectively attenuate stored information along directions associated with newly arriving data while preserving informative components in orthogonal directions. Recent work has incorporated directional forgetting through oblique projection decompositions to avoid estimator windup in both discrete-time linear systems \cite{SCC.Zhu.Yu.ea2022} and continuous-time nonlinear parameter estimation \cite{SCC.Wen.Xing.ea2025}.

Existing alert MRE approaches in \cite{SCC.Lee.Shin.ea2019} and \cite{SCC.Pan.Aranovskiy.ea2019}, however, assume exact knowledge of a linear parameterization of the system. When such a linear parameterization is not available, a large nonlinear dictionary of candidate functions is used for identification, and the resulting MRE includes a residual error term that also depends on the historical behavior of the system. Due to the trajectory-dependent residual error, uniform bounds and Lyapunov-based convergence guarantees are more difficult to establish. Moreover, the resulting overparameterized representation contains many candidate functions that do not contribute to the true dynamics.

The true dynamics typically depend on only a small subset of the selected basis functions, leading to a sparse parameter vector. Sparse identification methods such as Sparse Identification of Nonlinear Dynamics (SINDy) recover governing equations from large candidate libraries using sparsity-promoting regression techniques \cite{SCC.Brunton.Proctor.ea2016}. However, these approaches are typically formulated as offline batch optimization problems and are therefore not well-suited for real-time applications or integration within adaptive control loops.

Recent works \cite{SCC.Istiqphara.Wahyunggoro.ea2024, SCC.MedaCampana.Sanchez.ea2025} have explored sparse identification using recursive or streaming estimation schemes, where model parameters are updated sequentially from incoming data using techniques such as recursive least squares combined with sparsity-promoting regularization, as well as streaming weak-SINDy formulations that construct regression features incrementally from data streams \cite{SCC.Russo.Laiu.ea2025}. While these methods enable real-time model discovery and maintain interpretable state-space representations, they are generally developed for system identification in open-loop or offline settings. As a result, they do not explicitly account for closed-loop stability constraints that arise when the learned models or parameter updates are embedded within adaptive control laws, where the estimation dynamics and control dynamics are coupled.

Motivated by these observations, this paper develops an online sparse identification framework for nonlinear control-affine systems that integrates a novel MRE approach with novel sparsity-promoting adaptive update laws. Incorporating $\ell_1$ regularization within a recursive least-squares adaptation setting is challenging because sparsity-promoting loss functions are often nonsmooth, requiring tools from nonsmooth analysis in the design and the stability analysis. In this paper, we show that $\ell_1$-regularized recursive least squares is realized in an online, continuous-time setting via a sliding mode-like update law.  The combination of $\ell_1$ regularization and directional forgetting yields a discontinuous update law, and the closed-loop system is formulated as a differential inclusion and analyzed using Clarke generalized gradients.

A preliminary version of this work~\cite{SCC.Satharasi.Ogri.ea2026a} developed a sparsity-promoting integral concurrent learning (SP-ICL) update law using recorded integral data. In the SP-ICL framework, the aggregated finite-impulse-response (FIR) integral filters provide a PE surrogate through stored data, ensuring the memory matrix remains full rank and avoids the singularity issues associated with continuous regressor-based variable forgetting. However, the SP-ICL approach in \cite{SCC.Satharasi.Ogri.ea2026a} was restricted to exact linear parameterizations and implementation of the FIR filters in the control loop can become computationally prohibitive as the size of the history stack increases.

% This paper replaces SP-ICL with a sparsity-promoting variable-rate directional forgetting (SP-VDF) framework, derived from a generalized gradient descent applied to an $\ell_1$-regularized cost functional. While grounded in the broader MRE architecture, the SP-VDF framework replaces the FIR filters in ICL with low-pass filters that are implemented efficiently in their state-space form. The transition to MRE necessitates a shift from scalar discounting, such as \cite{SCC.Satharasi.Ogri.ea2026a, SCC.Parikh.Kamalapurkar.ea2019} to a matrix-valued directional forgetting mechanism. 
This paper replaces the history-stack-based SP-ICL approach \cite{SCC.Satharasi.Ogri.ea2026a} with a sparsity-promoting variable-rate directional forgetting (SP-VDF) framework based on generalized gradient descent of an $\ell_1$-regularized cost functional. Both methods use the MRE architecture to evaluate filtered signals, but SP-VDF replaces the discrete data storage of ICL with a continuous-time update of the regressor matrix. Consequently, scalar discounting and data purging used with history stacks, as in \cite{SCC.Satharasi.Ogri.ea2026a, SCC.Parikh.Kamalapurkar.ea2019}, are replaced by a continuous matrix-valued directional forgetting mechanism that preserves parameter alertness while keeping the parameter estimates bounded.

This paper further extends prior work by incorporating bounded function approximation errors, whereas \cite{SCC.Cao.Schwartz.ea2000, SCC.Lee.Shin.ea2019, SCC.Satharasi.Ogri.ea2026a} assume exact linear parameterizations. In the linearly parameterizable (LP) setting, the scalar forgetting mechanisms impose stronger excitation requirements such as the subspace PE condition (see Section~\ref{sec:MRE_alert}) to maintain boundedness of the memory regressor (MR). In contrast, the proposed SP-VDF update dynamically adapts the forgetting rate based on the conditioning of the MR and requires only FE to guarantee MR boundedness for LP systems. For non-LP systems, bounded approximation errors introduce an additional residual term, and the subspace PE condition is required to establish stability of the resulting residual dynamics.

% \textit{Section Outline:} Section~\ref{sec:problemformulation} formally defines the problem, outlines assumptions on the system and controller, Section~\ref{sec}

\section{Problem Formulation}\label{sec:problemformulation}
Consider the nonlinear control-affine system given by
\begin{equation}\label{eq:dynamics}
    \dot{x} = f(x) + g(x)u, \quad x(t_{0}) = x_{0},
\end{equation}
where $x \in \mathbb{R}^n$ is the state vector and $u \in \mathbb{R}^m$ is the control input. The drift dynamics $f: \mathbb{R}^{n} \to\mathbb{R}^n$ are unknown, while the control effectiveness matrix $g:\mathbb{R}^{n}\to\mathbb{R}^{n\times m}$ is known.

\begin{assumption}\label{ass:bounds}
 The functions $f$ and $g$ are locally Lipschitz continuous.
\end{assumption}

Let $\mathcal{D} \coloneqq \{Y_j\}_{j=1}^p$ denote a dictionary of locally Lipschitz continuous basis functions. Let the regressor $Y: \mathbb{R}^{n} \to \mathbb{R}^{n \times p}$ be defined as $Y \coloneqq \begin{bmatrix}Y_1 & \cdots & Y_p \end{bmatrix}$, where the $j$-th column corresponds to the candidate basis function $Y_j: \mathbb{R}^{n} \to \mathbb{R}^n$. 

\begin{assumption}\label{ass:approx}
The unknown drift dynamics can be approximated on a prescribed compact set $\mathbb{X} \subset \mathbb{R}^{n}$ containing the origin in its interior using the dictionary $\mathcal{D}$ as
\begin{equation}\label{eq:f_approx}
    f(x) = Y(x)\theta + \epsilon(x),
\end{equation}
where $\theta \in \mathbb{R}^p$ is an unknown constant parameter vector and the function approximation error $\epsilon:\mathbb{R}^{n}\to \mathbb{R}^n$ satisfies $\sup_{x\in\mathbb{X}}\|\epsilon(x)\|\le \bar\epsilon$ for a sufficiently small constant (in the sense made precise in Theorem~\ref{thm:main}) $\bar\epsilon \ge 0$. Furthermore, the parameter vector $\theta$ is assumed to be sparse in the sense that it contains at most $s$ nonzero entries, where $s \ll p$. 
\end{assumption}
\begin{assumption}
\label{ass:theta_bounded}
The unknown parameter vector $\theta \in \mathbb{R}^p$ is contained in a known compact set 
\begin{equation}
    \Theta \coloneqq \{\theta \in \mathbb{R}^p \mid \|\theta\| \le r_{\theta}\},
\end{equation}
where $r_{\theta} > 0$ is a known constant.
\end{assumption}
Under Assumption~\ref{ass:approx}, the system dynamics \eqref{eq:dynamics} can be written as
\begin{equation}\label{eq:sparseDynamics}
    \dot{x} = Y(x)\theta + g(x)u + \epsilon(x).
\end{equation}

The control objective is to track a desired trajectory $x_d: \mathbb{R}_{\geq 0} \to \mathbb{R}^n$ that satisfies the following assumption.
\begin{assumption}\label{ass:trajectory}
The reference trajectory $x_d$ is continuously differentiable, and both $x_d$ and $\dot{x}_d$ are bounded.
\end{assumption}

\begin{assumption}\label{ass:pseudo}
The matrix $g(x)$ has full row rank for all $x\in\mathbb{X}$. Consequently, there exists a locally Lipschitz right pseudoinverse $g^+(x)$ satisfying $g(x)g^+(x)=I_n$ for all $x \in \mathbb{X}$.
\end{assumption}
\begin{remark}
Assumption~\ref{ass:pseudo}, while restrictive, simplifies the control design and focuses the discussion on the design of the sparsity-promoting update law by allowing algebraic cancellation of the known control-effectiveness matrix through a right pseudoinverse. The design of tracking controllers for systems that do not satisfy this condition is outside the scope of this paper.
\end{remark}

Let $\hat{\theta}\in\mathbb{R}^p$ denote the online estimate of $\theta$ and let $K = K^{\top} \in \mathbb{R}^{n \times n}, K \succ 0 $ be a constant feedback gain matrix. Under Assumption~\ref{ass:pseudo}, the input is designed using certainty-equivalence control as
\begin{equation}\label{eq:control}
u = g^+(x)\big(\dot{x}_d - Y(x)\hat{\theta} - Ke\big),
\end{equation}
where $e \coloneqq x - x_d$ is the tracking error.
Substituting \eqref{eq:control} into \eqref{eq:sparseDynamics} yields the tracking error dynamics
\begin{equation}\label{eq:closed_error}
\dot e = -Ke + Y(x)\tilde{\theta} + \epsilon(x),
\end{equation}
where $\tilde{\theta} \coloneqq \theta - \hat{\theta}$ denotes the parameter estimation error.

The regressor matrix $Y$ is constructed using a large dictionary of candidate basis functions. Consequently, many dictionary elements may be redundant or weakly informative. For such regressors, gradient and least-squares-based update laws tend to distribute parameter energy across correlated basis functions, which can lead to large parameter norms and reduced interpretability of the identified model \cite{SCC.Brunton.Proctor.ea2016}. The objective of this paper is to design an online parameter estimator for the system \eqref{eq:sparseDynamics} that identifies a sparse parameter vector $\theta$ using state and input memory generated during system operation. To that end, the following section presents a brief overview of a memory-based adaptive controller.

\section{Memory-Based Regression: A Brief Review}\label{subsec:df_icl}
Memory-based learning in the presence of disturbances and/or function approximation errors requires an affine regression equation (AfRE) relating measurable signals to the unknown parameters. 
\subsection{Affine Regression Equation}
For the nonlinear system \eqref{eq:sparseDynamics}, an AfRE of the form
\begin{equation}\label{eq:filtered_regression}
u_f = Y_f\theta + d_f,
\end{equation}
is developed, where $u_f \in \mathbb{R}^n$ and $Y_f \in \mathbb{R}^{n\times p}$ are computable from measured data, and $d_f \in \mathbb{R}^n$ collects the filtered approximation error. The AfRE can be constructed using FIR or infinite-impulse-response (IIR) filters.

\subsubsection{Infinite-Impulse-Response Filtering}
Let $h$ denote a strictly proper exponentially stable linear time-invariant (LTI) system. Applying $h$ to both sides of \eqref{eq:sparseDynamics} yields
\begin{equation}
h[\dot{x}] = h[Y(x)]\,\theta + h[g(x)u] + h[\epsilon(x)],
\end{equation}
where $h[y](t)$ denotes the output of the LTI system, with input $y$, starting from zero initial state, and evaluated at time $t$. The filtered signals in \eqref{eq:filtered_regression} can then be defined as $u_f \coloneqq h[\dot{x}-g(x)u]$, $Y_f \coloneqq h[Y(x)]$, and $d_f \coloneqq h[\epsilon(x)]$.

A common choice is the first-order stable filter $h(s)=\frac{\varrho}{s+\varrho}$, where $\varrho>0$ and $s$ is the Laplace variable \cite{SCC.Ioannou.Sun1996}. This choice yields the filter commonly used in composite adaptive control \cite{SCC.Pan.Shi.ea2024} as
\begin{align}
    \dot{u}_{f} &= -\varrho\,u_{f} + \varrho\big(\dot{x}-g(x)u\big),
    & u_{f}(t_{0}) &= 0, \label{eq:uf_filter}\\
    \dot{Y}_{f} &= -\varrho\,Y_{f} + \varrho\,Y(x),
    & Y_{f}(t_{0}) &= 0, \label{eq:Yf_filter}\\
    \dot{d}_{f} &= -\varrho\,d_{f} + \varrho\,\epsilon(x),
    & d_{f}(t_{0}) &= 0. \label{eq:df_filter}
\end{align}
Note that under Assumption~\ref{ass:approx}, and provided that $x(t)\in\mathbb{X}$ for all $t\ge t_0$, there exists a constant $\bar{d} \ge 0$ such that $\|d_f(t)\|\le \bar{d}$ for all $t\ge t_{0}$.

\subsubsection{Finite-Impulse-Response Filtering}
Let $\Delta_t>0$ denote a window length. Integrating \eqref{eq:sparseDynamics} over $[t-\Delta_t,t]$ and applying the Fundamental Theorem of Calculus yields, for all $t \ge t_0 + \Delta_t$,
\begin{equation}\label{eq:int_reg}
x(t)-x(t-\Delta_t) = \int_{t-\Delta_t}^{t}Y(x(\tau))\,\mathrm{d}\tau\,\theta + \int_{t-\Delta_t}^{t}g(x(\tau))u(\tau)\,\mathrm{d}\tau + \int_{t-\Delta_t}^{t}\epsilon(x(\tau))\,\mathrm{d}\tau .
\end{equation}
For $t \geq t_{0} + \Delta_t$, the signals in \eqref{eq:filtered_regression} can thus be defined as $u_f(t) \coloneqq x(t)-x(t-\Delta_t) - \int_{t-\Delta_t}^{t}g(x(\tau))u(\tau)\,\mathrm{d}\tau$, $Y_f(t) \coloneqq \int_{t-\Delta_t}^{t}Y(x(\tau))\,\mathrm{d}\tau$, and $d_f(t) \coloneqq \int_{t-\Delta_t}^{t}\epsilon(x(\tau))\,\mathrm{d}\tau$. The signals are defined to be zero for $t_{0} \leq t < t_0 + \Delta_t$. Note that, under Assumption~\ref{ass:approx} and provided $x(\tau) \in \mathbb{X}$ for all $\tau \in [t-\Delta_t, t]$, $\|d_{f}(t)\|\le\int_{t-\Delta_t}^{t}\|\epsilon(x(\tau))\|\,\mathrm{d}\tau\le \Delta_t\bar\epsilon \eqqcolon \bar{d}$. The FIR filter was introduced in \cite{SCC.Parikh.Kamalapurkar.ea2019} under the \emph{integral concurrent learning} moniker.

\subsection{Memory Regressor Extension - Classical Approaches}\label{sec:mre}

Incorporating memory into the adaptive update laws is necessary to learn unknown parameters without PE. Premultiplying \eqref{eq:filtered_regression} by $Y_f^\top(t)$ yields
\begin{equation}
Y_f^\top(t)u_f(t) = Y_f^\top(t)Y_f(t)\theta + Y_f^\top(t)d_f(t). \label{eq:premultiplied_regression}
\end{equation}
When $p>n$, the matrix $Y_f^\top(t)Y_f(t)$ is generally rank deficient, and therefore, the AfRE cannot be uniquely solved for $\theta$.

The objective of MRE is to define signals $\mathcal{U}:[t_{0}, \, \infty) \to \mathbb{R}^{p}$ and $\mathcal{Y}:[t_{0}, \, \infty) \to \mathbb{R}^{p \times p}$ such that
\begin{equation}
\mathcal U(t)=\mathcal Y(t)\theta + \mathcal{E}(t),
\label{eq:memory_regression}
\end{equation}
where, unlike the instantaneous Gramian $Y_f^\top(t)Y_f(t)$, the MR $\mathcal Y(t)$ may become full rank under suitable interval excitation conditions (see, e.g., \cite{SCC.Chowdhary2010, SCC.Jha.Roy.ea2019,SCC.Pan.Shi.ea2024}). 

One approach to construct the MRE is to use saturated integrators. Define
\begin{align}
\dot{\mathcal U}(t) &=
\begin{cases}
Y_f^\top(t)u_f(t), & \|\mathcal Y(t)\|\le s_{1} \text{ and } \|\mathcal U(t)\|\le s_{2}, \\
0, & \text{otherwise},
\end{cases}
\\
\dot{\mathcal Y}(t) &=
\begin{cases}
Y_f^\top(t)Y_f(t), & \|\mathcal Y(t)\|\le s_{1} \text{ and } \|\mathcal U(t)\|\le s_{2}, \\
0, & \text{otherwise},
\end{cases}
\\ 
\dot{\mathcal{E}}(t) &=
\begin{cases}
Y_f^\top(t)d_f(t), & \|\mathcal Y(t)\|\le s_{1} \text{ and } \|\mathcal U(t)\|\le s_{2}, \\
0, & \text{otherwise},
\end{cases}
\end{align}
where $s_{1}, s_{2}>0$ are saturation thresholds. With initial conditions $\mathcal U(t_0)=0$, $\mathcal Y(t_0)=0$, and $\mathcal E(t_0)=0$, the regression relation \eqref{eq:memory_regression} is preserved for all $t\ge t_0$. The saturation ensures that $\mathcal{U}$, $\mathcal{Y}$, and $\mathcal{E}$ remain bounded, independent of the system trajectories. For design and analysis of MRE-based update laws and for other approaches to construct the MRE, see \cite{SCC.Jha.Roy.ea2019, SCC.Pan.Shi.ea2024, SCC.Chowdhary2010}. Classical MRE techniques yield an MR that is uniformly positive definite on $[T,\infty)$ for some $T> t_0$ under the FE condition.

\begin{definition}[Uniform PE]
\label{def:PE}
The filtered regressor $Y_f:[t_0,\infty)\to\mathbb{R}^{n\times p}$ is said to be \emph{uniformly persistently exciting} (u-PE) if there exist constants $T>0$ and $\gamma>0$, independent of $t_0$, $x(t_0)$, and $\hat{\theta}(t_0)$, such that for all $t>t_0$,
\begin{equation}
\int_{t}^{t+T} Y_f^\top(\tau)Y_f(\tau)\,\mathrm{d}\tau \succeq \gamma I_p.
\label{eq:PE}
\end{equation}
\end{definition}
\begin{definition}[Uniform FE]\label{def:FE}
The filtered regressor $Y_f:[t_0, \, \infty)\to\mathbb{R}^{n\times p}$ is said to be \emph{uniformly finitely exciting} (u-FE) over $[t_0,t_1]$ if there exists a constant $\gamma>0$, independent of $t_0$, $x(t_0)$, and $\hat{\theta}(t_0)$, such that
\begin{equation}
\int_{t_0}^{t_1} Y_f^\top(\tau)Y_f(\tau)\,\mathrm{d}\tau
\succeq
\gamma I_p.
\label{eq:FE}
\end{equation}
\end{definition}
\begin{assumption}\label{ass:FE}
The filtered regressor $Y_f:[t_0, \, \infty)\to\mathbb{R}^{n\times p}$ is u-FE over some interval $[t_0,t_{1}]$.
\end{assumption}
While the constructions in Section~\ref{sec:mre} yield bounded MRE signals and a uniformly positive definite MR under Assumption \ref{ass:FE}, the resulting update laws cannot track unexpected changes in the parameters due to the need to clear the memory terms corresponding to the previous values of the parameters. A forgetting mechanism is thus needed to keep the MRE alert to unexpected changes in the parameters.
%In the next section, we detail directional forgetting \cite{SCC.Lee.Shin.ea2019}, where information along directions excited by newly arriving regressors is forgotten while preserving the information along directions that are not excited.

\section{Memory Regressor Extension with Alertness}\label{sec:MRE_alert}
% Directional forgetting allows outdated information in the MR to be selectively removed when new excitation becomes available. However, forgetting alone does not guarantee that the stored information remains sufficiently informative for parameter estimation. To ensure parameter convergence without imposing the PE condition, historical data is integrated into the MR $\mathcal{Y}$. The construction process for $\mathcal{Y}$ must satisfy two conflicting objectives: it must be uniformly positive definite and uniformly bounded (to prevent estimator windup). The directional forgetting mechanism developed in this section achieves this balance.

% While MRE schemes with scalar forgetting factors can achieve parameter convergence under excitation conditions weaker than PE \cite{SCC.Jha.Roy.ea2019, SCC.Pan.Shi.ea2024, SCC.Parikh.Kamalapurkar.ea2019}, we show in the next section that uniform variable forgetting cannot simultaneously satisfy positivity and boundedness under FE.
A variable-forgetting scheme is developed in \cite{SCC.Pan.Aranovskiy.ea2019} to guarantee alertness in MRE-based adaptive control under FE, rather than PE. However, as shown in the following section, the MRE construction in \cite{SCC.Pan.Aranovskiy.ea2019} cannot guarantee boundedness of the MR.
\subsection{Loss of Boundedness in Uniform Forgetting Schemes}\label{sec:uniform_forgetting}
 In \cite{SCC.Pan.Aranovskiy.ea2019}, the MRE signals $\mathcal{Y}$ and $\mathcal{U}$ are updated according to
\begin{align} 
\dot{\mathcal{Y}} &= -\beta(t)\mathcal{Y} + Y_f^\top(t) Y_f(t), \quad \mathcal{Y}(t_0) = 0_{p \times p}, \label{eq:uniform_Y} \\ 
\dot{\mathcal{U}} &= -\beta(t)\mathcal{U} + Y_f^\top(t) u_f(t), \quad \mathcal{U}(t_0) = 0_{p \times 1}. \label{eq:uniform_U}
\end{align}
The forgetting rate $\beta$ is selected as a function of $\lambda_{\min}(\mathcal{Y})$ to ensure alertness and avoid the loss of convergence even when excitation is weak. Given thresholds $0<\underline{y}_1<\underline{y}_2$ and a maximum forgetting rate $\beta_{\max}>0$, the alertness forgetting factor in \cite{SCC.Pan.Aranovskiy.ea2019} is defined by
\begingroup
\begin{equation}\label{eq:beta_a}
\beta_{a}(\mathcal{Y}) = \frac{\beta_{\max}}{2}\left(\min\!\left(1,\max\!\left(-1,
\frac{\lambda_{\min}(\mathcal{Y})-\frac{\underline{y}_1+\underline{y}_2}{2}}{\frac{\underline{y}_2-\underline{y}_1}{2}}\right)
\right)+1 \right),
\end{equation}
with $\beta(t)=\beta_a(\mathcal{Y}(t))$.

While the time-varying forgetting parameter in \eqref{eq:beta_a} maintains alertness under FE by preventing the decay of stored information when excitation weakens and forgetting aggressively when the excitation is strong, the following example shows that it may lead to unboundedness of the MR, even when $Y_f$ is essentially bounded and FE. 
%Specifically, consider a regressor that is finitely exciting over a finite interval, but subsequently loses excitation on a subspace of the parameter space. Since the forgetting rate in \eqref{eq:beta_a} is determined solely by the smallest eigenvalue of $\mathcal{Y}$, a single poorly excited direction can force $\beta$ to zero, regardless of the excitation level in the remaining directions, which can lead to instability as shown in the following example.

% \begin{example}[Unboundedness under Partial Excitation]\label{ex:counterexample}
% Let the filtered regressor be given by
% \begin{equation}
% Y_f(t) = \begin{bmatrix} 1 + \sin(t)e^{-t} & 1 + \cos(3t)e^{-t} \end{bmatrix}.
% \end{equation}
% It is straightforward to verify that this regressor is finitely exciting over the interval $[0, 1]$. Consider the corresponding MR, constructed using \eqref{eq:uniform_Y} and \eqref{eq:beta_a} with $\beta_{\max} = 10$, $\underline{y}_1 = 0.05$, and $\underline{y}_2 = 0.2$. As illustrated by the numerical simulation in Figure~\ref{fig:uf_windup}, the lack of persistent excitation in the orthogonal direction causes $\lambda_{\min}(\mathcal{Y})$ to decay to the threshold $\underline{y}_1 = 0.05$. This forces the scalar forgetting factor $\beta$ to collapse near zero. Due to the resulting lack of sufficient forgetting, the maximum eigenvalue grows unbounded, i.e., $\lambda_{\max}(\mathcal{Y}) \to \infty$. 
% \end{example}

\begin{example}[Unboundedness under Partial Excitation]\label{ex:counterexample}
Let the filtered regressor be given by
\begin{equation}
Y_f(t) = \begin{bmatrix} 1 + \sin(t)\varpi(t) & 1 + \cos(3t)\varpi(t) \end{bmatrix},
\end{equation}
where $\varpi(t) \coloneqq \max(0, 1-t)$. It is straightforward to verify that this regressor is FE over the interval $[0, 1]$. Consider the corresponding MR, constructed using \eqref{eq:uniform_Y} and \eqref{eq:beta_a} with $\beta_{\max} = 10$, $\underline{y}_1 = 0.05$, and $\underline{y}_2 = 0.2$. For $t \ge 1$, the regressor becomes exactly the constant rank-1 vector $[1 \;\; 1]$. Consequently, $\lambda_{\min}(\mathcal{Y}(t))$ monotonically decreases until it reaches the threshold $\underline{y}_1 = 0.05$ at some finite time $t = T > 1$. At this instant, $\beta(T) = 0$. Because the residual excitation is exactly zero for $t \ge 1$, $\lambda_{\min}(\mathcal{Y}(t))$ remains exactly $0.05$ and $\beta(t)$ remains zero for all $t \ge T$. Due to the lack of forgetting, the constant excitation in the $[1 \;\; 1]$ direction yields $\lambda_{\max}(\mathcal{Y}(t)) \to \infty$. 
\end{example}

% \begin{figure}[t]
%     \centering
%     \input{figures/failure_example}
%     \caption{Simulation of Example~\ref{ex:counterexample}. As the minimum eigenvalue decays to the threshold $\underline{y}_1$, the scalar forgetting factor $\beta(t)$ is forced into a limit cycle near zero. Consequently, the lack of sufficient forgetting causes the maximum eigenvalue to grow unbounded.}
%     \label{fig:uf_windup}
% \end{figure}

\subsection{Bounded Variable Forgetting}\label{sec:bounded_forgetting}
To keep the MR bounded under variable forgetting, we introduce the thresholds $\overline{y}_1 < \overline{y}_2$, satisfying $0 < \underline{y}_1 < \underline{y}_2 \le \gamma < \overline{y}_1 < \overline{y}_2$, where $\gamma$ is the FE constant in Definition~\ref{def:FE}. We define a bounding forgetting factor $\beta_b(\mathcal{Y})$ to enforce maximal forgetting when the maximum eigenvalue of the MR exceeds $\overline{y}_2$ as
\begin{equation}
\beta_b(\mathcal{Y}) =\frac{\beta_{\max}}{2}\left(\min\left(1,\max\left(-1,\frac{\lambda_{\max}\left(\mathcal{Y}\right) - \frac{\overline{y}_1 + \overline{y}_2}{2}}{\frac{\overline{y}_2 - \overline{y}_1}{2}}\right)\right) + 1\right).
\label{eq:beta_b}
\end{equation}
\endgroup
We then define the forgetting factor using \eqref{eq:beta_a} and \eqref{eq:beta_b} as
\begin{equation}\label{eq:beta_composite}
\beta(t) = \max \Big( \beta_a\big(\mathcal{Y}(t)\big), \, \beta_b\big(\mathcal{Y}(t)\big) \Big),
\end{equation}
to ensure the boundedness of $\mathcal{Y}$ (see Theorem~\ref{thm:Y_bounded}).

Note that the boundedness mechanism in \eqref{eq:beta_composite} may render the MR uninformative under FE. Indeed, if $\lambda_{\max}(\mathcal{Y})$ grows, then the bounding mechanism enforces $\beta=\beta_{\max}$, yielding the decay term $-\beta_{\max}\|\mathcal{Y}\|^2$.  Because this decay acts uniformly across all directions, it attenuates information in weakly excited directions as well as strongly excited ones. Consequently, if $Y_f$ is not PE, $\lambda_{\min}(\mathcal{Y})$ may converge to zero, causing $\mathcal{Y}$ to lose positive definiteness. To ensure uniform positive definiteness of the MR under a condition weaker than PE, we use directional forgetting, first introduced in \cite{SCC.Cao.Schwartz.ea2000}. 

\subsection{ Variable-Rate Directional Forgetting}\label{subsec:VRDF}
Using the time-varying forgetting factor defined in \eqref{eq:beta_composite}, the MRE signals $\mathcal{Y}$ and $\mathcal{U}$ are constructed according to 
\begin{align}
\dot{\mathcal{Y}}(t) =
\begin{cases}
Y_f^\top(t)Y_f(t),
& t \in [t_0, t_1],
\\
Y_f^\top(t)Y_f(t)
-\beta(t)
\dfrac{\mathcal{Y}(t)Y_f^\top(t)Y_f(t)\mathcal{Y}(t)}
{\tr\big(Y_f(t)\mathcal{Y}(t)Y_f^\top(t)\big)},
& t > t_1,
\end{cases}
\label{eq:Y_dot}
\\
\dot{\mathcal{U}}(t) =
\begin{cases}
Y_f^\top(t)u_f(t),
& t \in [t_0, t_1],
\\
Y_f^\top(t)u_f(t)
-\beta(t)
\dfrac{\mathcal{Y}(t)Y_f^\top(t)Y_f(t)}
{\tr\big(Y_f(t)\mathcal{Y}(t)Y_f^\top(t)\big)}\,\mathcal{U}(t),
& t > t_1,
\end{cases}
\label{eq:U_dot}
\end{align}
where the forgetting terms in \eqref{eq:Y_dot}--\eqref{eq:U_dot} are defined to be zero when $Y_f(t)=0$ for $t>t_1$. The differential equations are initialized at $t=t_0$ with $\mathcal{Y}(t_0)=0_{p\times p}$ and $\mathcal{U}(t_0)=0_{p\times1}$.

To guarantee boundedness and uniform positive definiteness of the MR, the following assumption is needed.
\begin{assumption}\label{ass:PE}
There exists a fixed $q$-dimensional subspace $\phi \subseteq \mathbb{R}^p$ within which the regressor $Y_{f}$ is u-PE. Specifically, there exists an orthogonal matrix $U = \begin{bmatrix} U_1 & U_2 \end{bmatrix}$, where $U_1 \in \mathbb{R}^{p \times q}$ spans $\phi$ and $U_2 \in \mathbb{R}^{p \times (p-q)}$ spans the orthogonal complement $\phi^\perp$, such that $Y_{f}(t) U_2 = 0$ for all $t \ge t_1$ and there exist constants $T > 0$ and $\alpha > 0$, independent of $t_0$, $x(t_0)$, and $\tilde{\theta}(t_0)$, such that for all $t \ge t_1$,
\begin{equation}
  \int_{t}^{t+T} U_1^\top Y_{f}^\top(\tau)Y_{f}(\tau) U_1 \, d\tau \succeq \alpha I_q.
\end{equation}
\end{assumption}
\begin{remark}
Note that u-PE of the filtered regressor on a fixed subspace with $q < p$ is weaker than u-PE on $\mathbb{R}^{p}$. Furthermore, while not assumed explicitly in \cite{SCC.Lee.Shin.ea2019}, u-PE on a fixed subspace $\phi$ is needed for the results in \cite{SCC.Lee.Shin.ea2019} to hold. In the following, we provide new proofs of the bounds on the MR that clarify the need for u-PE on a fixed subspace. We also extend the arguments in \cite{SCC.Lee.Shin.ea2019} to the present case where the regressor is matrix-valued rather than vector-valued, and the function being approximated is not linear in the parameters.   
\end{remark}

\subsection{Lower and Upper Bounds on the Memory Regressor}

In this section, we derive lower and upper bounds on the eigenvalues of the MR under suitable excitation conditions on the filtered regressor, utilizing the following technical lemmas.

\begin{lemma}\label{lem:eigderiv_max}
If $M:[0,\infty)\to\mathbb{R}^{n\times n}$ is an absolutely continuous symmetric-matrix-valued function, then $t\mapsto\lambda_{\max}(M(t))$ is absolutely continuous, and for almost every $t$ there exists a unit eigenvector $v(t)$ corresponding to $\lambda_{\max}(M(t))$ such that
\begin{equation}\label{eq:eig_derivative_max}
\frac{\mathrm{d}}{\mathrm{d}t}\lambda_{\max}(M(t))
\le
v(t)^{\top}\dot{M}(t)\,v(t).
\end{equation}
\end{lemma}
\begin{proof}
By the Rayleigh--Ritz characterization of the largest eigenvalue, $\lambda_{\max}(M(t)) = \max_{\|u\|=1} u^\top M(t)\, u$. For any $h>0$, let $v(t+h)$ be a unit eigenvector achieving the maximum at $t+h$, so that $\lambda_{\max}(M(t+h)) = v(t+h)^\top M(t+h)\, v(t+h)$. Since $v(t+h)$ is a unit vector and $\lambda_{\max}(M(t))$ is the maximum of the Rayleigh quotient at time $t$, we have $\lambda_{\max}(M(t)) \ge v(t+h)^\top M(t)\, v(t+h)$.

Subtracting these two inequalities yields
\begin{equation}
\lambda_{\max}(M(t+h)) - \lambda_{\max}(M(t))\le v(t+h)^\top\bigl[M(t+h)-M(t)\bigr]v(t+h).
\end{equation}
Dividing by $h>0$ and taking the limit as $h\to 0^+$, Danskin's Theorem for the directional derivative of a maximum \cite[Appendix B]{SCC.Basar.Bernhard1995} gives \eqref{eq:eig_derivative_max} for almost every $t$.
\end{proof}
% \begin{proof}
% By the Rayleigh--Ritz characterization of the largest eigenvalue, $\lambda_{\max}(M(t)) = \max_{\|u\|=1} u^\top M(t)\, u$. For any $h>0$, let $v(t+h)$ be a unit eigenvector achieving the maximum at $t+h$, so that $\lambda_{\max}(M(t+h)) = v(t+h)^\top M(t+h)\, v(t+h)$. Since $v(t+h)$ is a unit vector and $\lambda_{\max}(M(t))$ is the maximum of the Rayleigh quotient at time $t$, we have $\lambda_{\max}(M(t)) \ge v(t+h)^\top M(t)\, v(t+h)$.

% Subtracting these two inequalities yields
% \begin{multline}
% \lambda_{\max}(M(t+h)) - \lambda_{\max}(M(t))
% \\\le 
% v(t+h)^\top\bigl[M(t+h)-M(t)\bigr]v(t+h).
% \end{multline}
% Dividing by $h>0$ gives
% \begin{multline}\label{eq:eig_diff_quotient}
% \frac{\lambda_{\max}(M(t+h)) - \lambda_{\max}(M(t))}{h} \\\le v(t+h)^\top\left[\frac{M(t+h)-M(t)}{h}\right]v(t+h).
% \end{multline}
% Because the unit sphere is compact, there exists a sequence $h_k\to0^+$ such that $v_{h_k}\to v$ for some unit vector $v$. Since $M(t+h_k)v_{h_k}=\lambda_{\max}(M(t+h_k))v_{h_k}$, the continuity of $M$ and $\lambda_{\max}(M(\cdot))$ implies $M(t)v=\lambda_{\max}(M(t))v$. Thus, $v$ is a unit eigenvector associated with $\lambda_{\max}(M(t))$. Since $M$ is differentiable at $t$, $\frac{M(t+h_k)-M(t)}{h_k}\to\dot{M}(t)$, while differentiability of $\lambda_{\max}(M(\cdot))$ at $t$ gives $\frac{\lambda_{\max}(M(t+h_k))-\lambda_{\max}(M(t))}{h_k}\to\frac{\mathrm{d}}{\mathrm{d}t}\lambda_{\max}(M(t))$. Taking the limit in \eqref{eq:eig_diff_quotient} yields \eqref{eq:eig_derivative_max}.
% \end{proof}

\begin{lemma}\label{lem:full_excitation_tv}
Consider the $q$-dimensional directional-forgetting dynamics
\begin{equation}\label{eq:lemmaDyn}
  \dot A(t) = Q^\top(t)Q(t) - \beta(t)\,\frac{A(t)\,Q^\top(t)Q(t)\,A(t)}{\operatorname{tr}(Q^\top(t)Q(t)A(t))},
\end{equation}
where $A(t_1)\succ 0$ and $0 \le \beta(t) \le \bmax$ for all $t\geq t_1$, and the forgetting term is defined to be zero when $Q(t)=0$. If $Q$ is essentially bounded and u-PE on $[t_1,\infty)$ with constants $T > 0$ and $\gamma > 0$, then there exists a constant $\alpha_1 > 0$ such that
\begin{equation}
  A(t) \succeq \alpha_1 I_q,
\end{equation}
for all $t$ in the interval of existence of the solution of \eqref{eq:lemmaDyn} starting from $(t_1, A(t_1))$.
\end{lemma}
\begin{proof}
Local existence of solutions to \eqref{eq:lemmaDyn} is guaranteed by the Carath\'{e}odory existence theorem \cite[Section 1.4]{SCC.Coddington.Levinson1955}. Let $[t_1,t^{**})$ be the maximal interval of existence of the solution of \eqref{eq:lemmaDyn} starting from $(t_1, A(t_1))$. Since the ODE preserves symmetry, $A(t)$ is symmetric for all $t\in[t_1,t^{**})$. Let $[t_1,t^\ast)$ be the maximal interval on which $A(t)\succ 0$. This interval is nonempty since $A(t_1)\succ 0 $. On this interval, $A^{1/2}(t)$ exists. We will show that $A(t)\ge\alpha_1 I_q \succ 0$ uniformly on $[t_1,t^\ast)$, which forces $t^\ast=t^{**}$. 

Fix any unit vector $v\in\mathbb{R}^q$, let $S\coloneqq Q^\top Q$, and define the scalar function $a_v(t)=v^\top A(t)v$. Its derivative is
\begin{equation}
  \dot a_v(t) = v^\top S(t)v - \beta(t)\,\frac{v^\top A(t)S(t)A(t)v}{\operatorname{tr}(S(t)A(t))}.
\end{equation}
Let $H(t)=A^{1/2}(t)S(t)A^{1/2}(t)\succeq 0$. Bounding the quadratic form gives 
\begin{multline}
    v^\top A(t)S(t)A(t)v = (A^{1/2}(t)v)^\top H(t)(A^{1/2}(t)v) \\\le \lambda_{\max}(H(t))\,v^\top A(t)v \le \operatorname{tr}(H(t))\,a_v(t) = \operatorname{tr}(S(t)A(t))\,a_v(t).
\end{multline} Therefore,
\begin{equation}
  \dot a_v(t) \ge -\beta(t) a_v(t) + v^\top S(t)v \ge -\bmax a_v(t) + v^\top S(t)v.
\end{equation}
Applying the Comparison Lemma \cite[Lemma 3.4]{SCC.Khalil2002} and the variation-of-constants formula on any interval $[t_1,t^\ast)$ gives
\begin{equation}
  a_v(t) \ge e^{-\bmax(t-s)}a_v(s) + \int_s^t e^{-\bmax(t-\tau)}v^\top S(\tau)v\,d\tau.
\end{equation}
If $t\in[t_1,t_1+T]$, choosing $s=t_1$ yields $a_v(t)\ge e^{-\bmax T}\lambda_{\min}(A(t_1))$. If $t\ge t_1+T$, choosing $s=t-T$ gives $a_v(t) \ge \int_{t-T}^t e^{-\bmax(t-\tau)}v^\top S(\tau)v\,d\tau \ge e^{-\bmax T}v^\top\!\left(\int_{t-T}^t S(\tau)\,d\tau\right)\!v \ge e^{-\bmax T}\gamma$.
Hence, for every unit vector $v$ and every $t\in[t_1,t^\ast)$, $v^\top A(t)v\ge \alpha_1$, where $\alpha_1 = e^{-\bmax T}\min\{\lambda_{\min}(A(t_1)),\gamma\} > 0$. Because $\alpha_1$ is independent of $t^\ast$ and $A$ is continuous, we conclude that $t^\ast=t^{**}$.
\end{proof}

We now use the above lemmas to show boundedness and positive definiteness of the MR under the assumption that the regressor is essentially bounded, FE, and PE over a fixed subspace.

\begin{theorem}\label{thm:Y_bounded}
Suppose Assumptions~\ref{ass:bounds}--\ref{ass:PE} hold and there exist $\tau,\overline{y}_f \geq 0$ such that  $\|Y_f(t)\|\leq \overline{y}_f$, for almost all $t\in[t_0,\tau)$, then there exist constants $\overline{y} > \underline{y} > 0$ such that
\begin{equation}
\mathcal{Y}(t) \preceq \overline{y}\, I_p, \quad \forall t\in[t_0,\tau),  \label{eq:Y_upper}
\end{equation}
In addition, if $\tau > t_1$, then
\begin{equation}
    \underline{y}\, I_p \preceq \mathcal{Y}(t), \quad \forall t\in[t_1,\tau). \label{eq:Y_lower}
\end{equation}
\end{theorem}
\begin{proof}
In the case where $\tau > t_1$, boundedness of $Y_{f}$ ensures existence and uniqueness of solutions of \eqref{eq:Y_dot} on $[t_0,t_1]$. Because $\mathcal{Y}(t_1) = \int_{t_0}^{t_1} Y_{f}^\top(\tau)Y_{f}(\tau)\,d\tau \succeq \gamma I_p \succ 0$ by Assumption~\ref{ass:FE}, the matrix $\mathcal{Y}(t_1)$ is positive definite, activating the forgetting branch at $t_1$. Local existence of solutions of the forgetting branch of \eqref{eq:Y_dot} starting from $(t_1, \mathcal{Y}(t_1))$ is guaranteed by the Carath\'{e}odory existence theorem \cite[Section 1.4]{SCC.Coddington.Levinson1955}. Let $[t_1,t^{**})\subseteq [t_1,\tau)$ be the maximal interval of existence of the forgetting branch of \eqref{eq:Y_dot} starting from $(t_1, \mathcal{Y}(t_1))$. Continuity of $\mathcal{Y}$, along with $\mathcal{Y}(t_1) \succ 0$, guarantees a maximal interval $[t_1,t^\ast) \subseteq [t_1,t^{**})$ on which $\mathcal{Y}(t) \succ 0$.

To establish the uniform lower bound in \eqref{eq:Y_lower}, let $U = [U_1 \;\; U_2] \in \mathbb{R}^{p \times p}$ be the fixed orthogonal matrix from Assumption~\ref{ass:PE}. Let $Q(t) = Y_{f}(t) U_1$ and $S(t) = Q^\top(t)Q(t)$. Decompose the MR as
\begin{equation}\label{eq:decomp}
  U^\top\mathcal{Y}(t)U = \begin{bmatrix} A(t) & C(t) \\ C^\top(t) & D(t) \end{bmatrix},
\end{equation}
where $A(t) \in \mathbb{R}^{q \times q}$, $C(t) \in \mathbb{R}^{q \times (p-q)}$, and $D(t) \in \mathbb{R}^{(p-q) \times (p-q)}$. Since the forgetting branch of \eqref{eq:Y_dot} is active on $[t_1,t^\ast)$, the corresponding block dynamics are
\begin{align}
  \dot A &= S - \beta(t)\frac{ASA}{\tr(SA)}, \label{eq:A_dynamics}\\
  \dot C &= -\beta(t)\frac{ASC}{\tr(SA)}, \\
  \dot D &= -\beta(t)\frac{C^\top SC}{\tr(SA)}.
\end{align}
Since $Q$ is essentially bounded and u-PE by Assumption~\ref{ass:PE}, Lemma~\ref{lem:full_excitation_tv} can be invoked to conclude that there exists $\alpha_1 > 0$ such that $A(t) \succeq \alpha_1 I_q$ for all $t \in [t_1, t^\ast)$.

The Schur complement $\Sigma(t) = D(t) - C^\top(t)A^{-1}(t)C(t)$ satisfies
\begin{equation}
  \dot\Sigma = \dot D - \dot C^\top A^{-1}C - C^\top\dot A^{-1}C - C^\top A^{-1}\dot C = C^\top A^{-1}SA^{-1}C \succeq 0.
\end{equation}
Consequently, $\Sigma(t) \ge \Sigma(t_1)$. Because $\mathcal{Y}(t_1) \succ 0$, $\Sigma(t_1)$ is positive definite; thus, there exists $\alpha_2 > 0$ such that $\Sigma(t) \succeq \alpha_2 I_{p-q}$ for all $t \in [t_1, t^\ast)$. Furthermore, $\dot D \preceq 0$, ensuring $D(t) \preceq D(t_1)$. Defining $B(t) = A^{-1}(t)C(t)$, we obtain
\begin{equation}
  \alpha_1 B^\top(t)B(t) \preceq B^\top(t)A(t)B(t) = C^\top(t)A^{-1}(t)C(t) = D(t) - \Sigma(t) \preceq D(t_1) - \Sigma(t_1).
\end{equation}
This confirms that $B$ is uniformly bounded on $[t_1,t^\ast)$, i.e., there exists $M > 0$ such that $\|B(t)\| \le M$ for all $t\in [t_1,t^\ast)$. The block factorization of $U^\top\mathcal{Y}(t)U$ shows that for every $z \in \R^p$,
\begin{equation}
  z^\top U^\top\mathcal{Y}(t)Uz \ge \min(\alpha_1,\alpha_2) \left\| \begin{bmatrix} I_{q} & B(t) \\ 0 & I_{p-q} \end{bmatrix} z \right\|^2 \ge \frac{\min(\alpha_1,\alpha_2)}{(1+M)^2} \|z\|^2.
\end{equation}
Hence, $\mathcal{Y}(t) \succeq \underline{y}I_p$ on $[t_1,t^\ast)$, where $\underline{y} = \min(\alpha_1,\alpha_2)/(1+M)^2 > 0$. Since $\underline{y}$ is independent of $t^\ast$ and $\mathcal{Y}$ is continuous, we conclude that $t^\ast = t^{**}$.

To establish the upper bound \eqref{eq:Y_upper} over $[t_1, t^{**})$, let $\lambda_{\mathcal{Y}}(t) = \lambda_{\max}(\mathcal{Y}(t))$. By Lemma~\ref{lem:eigderiv_max}, $\lambda_{\mathcal{Y}}$ is absolutely continuous, and for almost every $t$ there exists a unit eigenvector $w(t)$ such that $\dot{\lambda}_{\mathcal{Y}}(t) \le w^\top(t)\dot{\mathcal{Y}}(t)w(t)$. Substituting the forgetting dynamics yields
\begin{equation}\label{eq:lambda_dot}
  \dot{\lambda}_{\mathcal{Y}}(t) \le w^\top(t)Y_{f}^\top(t)Y_{f}(t)w(t)  - \beta(t)\frac{w^\top(t)\mathcal{Y}(t)Y_{f}^\top(t)Y_{f}(t)\mathcal{Y}(t)w(t)}{\tr(Y_{f}(t)\mathcal{Y}(t)Y_{f}^\top(t))}.
\end{equation}
Using $\mathcal{Y}(t)w(t) = \lambda_{\mathcal{Y}}(t)w(t)$ and $\tr(Y_{f}(t)\mathcal{Y}(t)Y_{f}^\top(t)) \le \lambda_{\mathcal{Y}}(t)\|Y_{f}(t)\|^2 \le \lambda_{\mathcal{Y}}(t) p \overline{y}_f^2$, the bound in \eqref{eq:lambda_dot} simplifies to
\begin{equation}
  \dot{\lambda}_{\mathcal{Y}}(t) \le \|Y_{f}(t)w(t)\|^2 \left( 1 - \frac{\beta(t)\lambda_{\mathcal{Y}}(t)}{p\overline{y}_f^2} \right).
\end{equation}
By the design of the time-varying gain, when $\lambda_{\mathcal{Y}}(t) \ge \overline{y}_2$, we have $\beta(t) = \bmax$. Defining $\lambda^\ast = p\overline{y}_f^2 / \bmax$, for almost all $t$ where $\lambda_{\mathcal{Y}}(t) > \lambda^\ast$ and $\lambda_{\mathcal{Y}}(t) \ge \overline{y}_2$, it follows that $\dot{\lambda}_{\mathcal{Y}}(t) \le 0$. Consequently, defining
\begin{equation}
  \overline{y} \coloneqq \max\left( \overline{y}_f^2(t_1-t_0),\; \overline{y}_2,\; \frac{p\,\overline{y}_f^2}{\bmax} \right),
\end{equation}
ensures that $\lambda_{\mathcal{Y}}(t) \leq \overline{y}$. Thus, $\mathcal{Y}(t) \preceq \overline{y}I_p$ for all $t \in [t_1, t^{**})$. Since the bounds \eqref{eq:Y_upper} and \eqref{eq:Y_lower} hold on the maximal interval of existence $[t_1, t^{**})$, the solution $\mathcal{Y}$ can be extended to the entire interval $[t_0,\tau)$ where $Y_f$ is bounded. Hence, $t^{**} \geq \tau$, and $\underline{y}I_p\preceq\mathcal{Y}(t)\preceq\overline{y}I_p$ for all $t\in [t_1,\tau)$. %If $Y_f$ is bounded on $[t_1,\infty)$ then the bounds in \eqref{eq:Y_upper} and \eqref{eq:Y_lower} hold for all $t\geq t_1$.
\end{proof}

\begin{remark}\label{rem:composite_design}
The proof of Theorem~\ref{thm:Y_bounded} shows that a constant forgetting parameter $\beta(t)\equiv\beta_{\max}$ is sufficient to guarantee both the upper and lower eigenvalue bounds after the u-FE interval. Consequently, a time-varying forgetting factor is not required for boundedness and positive definiteness of the MR. Its purpose is instead to improve alertness in response to changes in the parameters. 

%The alertness-based gain $\beta_a(\mathcal{Y}(t))$ reduces forgetting under weak excitation. However, using $\beta(t)=\beta_a(\mathcal{Y}(t))$ alone does not guarantee boundedness. For example, let $p=2$ and suppose $Y_{f}(t)=e_1^\top$ for all $t\ge t_1$. The minimum eigenvalue $\lambda_{\min}(\mathcal{Y}(t))=\mathcal{Y}_{22}(t)$ remains constant. As $\beta_a\to0$, the forgetting term vanishes, while $\dot{\mathcal{Y}}_{11}(t)\approx\|Y_{f}(t)e_1\|^2=1$, so $\mathcal{Y}_{11}(t)\to\infty$. The additional component $\beta_b$ prevents this behavior by activating only when $\lambda_{\max}(\mathcal{Y}(t))\ge\overline{y}_2$. Hence, $\beta(t)=\max\{\beta_a,\beta_b\}$ preserves the adaptive behavior of $\beta_a$ under weak excitation while ensuring the uniform upper bound established in Theorem~\ref{thm:Y_bounded}.
\end{remark}

\subsection{Bounds on the Recursive Residual}\label{sec:residual_bound}

Unlike the results in \cite{SCC.Lee.Shin.ea2019} where the linear parameterization is exact, the dynamics in this paper are not assumed to admit an exact linear parameterization. As such, we must establish uniform bounds on the recursive residual $\mathcal{E}(t)$. The residual $\mathcal{E} \coloneqq \mathcal{U} - \mathcal{Y}\theta$ satisfies
\begin{equation}\label{eq:E_dot}
\dot{\mathcal{E}}=
\begin{cases}
Y_f^{\top}(t)d_f(t), & \text{if } t \in [t_0, t_1],\\
Y_f^\top(t)d_f(t) - \beta(t)
\dfrac{\mathcal{Y}Y_f^\top(t)Y_f(t)}
{\operatorname{tr}\!\big(Y_f(t)\mathcal{Y}Y_f^\top(t)\big)}
\mathcal{E}, & t > t_1,
\end{cases}
\end{equation}
where the forgetting term in \eqref{eq:E_dot} is defined to be zero when $Y_f(t)=0$ for $t>t_1$. The update laws \eqref{eq:Y_dot} and \eqref{eq:U_dot} ensure that $\mathcal{Y}(t)$ remains uniformly bounded and uniformly positive definite under Assumptions~\ref{ass:FE} and \ref{ass:PE}.

\begin{theorem}\label{thm:E_bounded}
%Let Assumptions~\ref{ass:bounds}--\ref{ass:PE} hold. Suppose the filtered regressor $Y_f$ is essentially bounded, i.e., $\|Y_f(t)\| \le \overline{y}_f$ for some $\overline{y}_f > 0$.
If the hypotheses of Theorem \ref{thm:Y_bounded} hold and there exists $\bar{d}\geq 0$ such that $\|d_f(t)\| \le \bar{d}$ for almost all $t\in [t_0,\tau)$, then the residual $\mathcal{E}(t)$ is uniformly bounded for all $t\in [t_0,\tau)$.
\end{theorem}
\begin{proof}
Boundedness of $\mathcal{E}$ on $[t_0,\min(t_1,\tau)]$ follows trivially from the boundedness of $Y_f$ and $d_f$ and continuity of $\mathcal{E}$. In the case where $\tau > t_1$, by Theorem~\ref{thm:Y_bounded}, the MR satisfies $\mathcal{Y}(t) \succeq \underline{y}I_{p} \succ 0$ for all $t \ge t_1$. 
%In particular, $\operatorname{rank}(\mathcal{Y}(t))=p$, ensuring that the forgetting branch of \eqref{eq:Y_dot} and \eqref{eq:E_dot} remains persistently active for all $t \ge t_1$. RK: not needed anymore since we changed the update law.

Let $U = [U_1 \;\; U_2]$ be the orthogonal matrix from Assumption~\ref{ass:PE}, $Q = Y_f U_1$, and $S = Q^\top Q$. Apply the decomposition \eqref{eq:decomp} and partition the residual as $\mathcal{E}_1 = U_1^\top\mathcal{E}$, $\mathcal{E}_2 = U_2^\top\mathcal{E}$. Projecting \eqref{eq:E_dot} onto $\phi$ and utilizing $Y_f^\top = U_1 Q^\top$ yields the dynamics for the excited subspace as
\begin{equation}\label{eq:E1_dot}
  \dot{\mathcal{E}}_1 = -\beta(t)\,\frac{A(t)\,S(t)}{\operatorname{tr}(S(t)A(t))}\,\mathcal{E}_1 + Q^\top(t) d_f(t).
\end{equation}
To analyze the stability of \eqref{eq:E1_dot}, we introduce the coordinate transformation $\xi(t) \coloneqq A^{-1}(t)\mathcal{E}_1(t)$. Differentiating $\xi(t)$ yields
\begin{equation}\label{eq:zeta_dot_temp}
    \dot{\xi} = \dot{A}^{-1}\mathcal{E}_1 + A^{-1}\dot{\mathcal{E}}_1.
\end{equation}
Applying the matrix identity $\dot{A}^{-1} = -A^{-1}\dot{A}A^{-1}$ and the block dynamics $\dot{A} = S - \beta(t) \frac{A S A}{\operatorname{tr}(SA)}$ from \eqref{eq:A_dynamics}, the inverse dynamics are given by
\begin{equation}\label{eq:Ainverse_dot}
    \dot{A}^{-1} = -A^{-1} S A^{-1} + \beta(t) \frac{S}{\operatorname{tr}(SA)}.
\end{equation}
Substituting \eqref{eq:E1_dot} and \eqref{eq:Ainverse_dot} into \eqref{eq:zeta_dot_temp}, and using $\mathcal{E}_1 = A \xi$, gives
% \begin{multline}
%     \dot{\xi} = \left( -A^{-1} S A^{-1} + \beta(t) \frac{S}{\operatorname{tr}(SA)} \right) \mathcal{E}_1 \\+ A^{-1} \left( -\beta(t)\frac{A S}{\operatorname{tr}(SA)}\mathcal{E}_1 + Q^\top d_f \right) \\= -A^{-1} S A^{-1} \mathcal{E}_1 + A^{-1} Q^\top d_f.
% \end{multline}
% and using $\mathcal{E}_1 = A \xi$, the dynamics for the transformed state are independent of $\beta(t)$ as
\begin{equation}\label{eq:zeta_dot}
    \dot{\xi} = -A^{-1}(t) S(t) \xi + A^{-1}(t) Q^\top(t) d_f(t).
\end{equation}

To analyze the unforced system $\dot{\xi} = -A^{-1} S \xi$, consider the candidate Lyapunov function $V_\xi(t, \xi) \coloneqq \xi^\top A(t) \xi$. Because $\alpha_1 I_q \preceq A(t) \preceq \overline{y} I_q$ by Theorem~\ref{thm:Y_bounded}, $V_\xi$ is bounded as $\alpha_1 \|\xi\|^2 \le V_\xi(t,\xi) \le \overline{y} \|\xi\|^2$. The Lie derivative of $V_\xi$ along the system trajectories is given by
\begin{equation}\label{eq:V_zetadot}
    \dot{V}_\xi = -\xi^\top S \xi - \beta(t) \frac{\xi^\top A S A \xi}{\operatorname{tr}(SA)}.
\end{equation}
When $Q(t)=0$, the forgetting term is defined as zero. When $Q(t) \neq 0$, $A(t) \succ 0$ implies $\operatorname{tr}(S(t)A(t)) > 0$. Because $\beta(t) \geq 0$ and $A(t)S(t)A(t) \succeq 0$ (since $S(t) = Q^\top(t) Q(t) \succeq 0$), the second term in \eqref{eq:V_zetadot} is non-positive in either case. Thus, $\dot{V}_\xi \le -\xi^\top S(t) \xi$. Since $V_\xi$ is bounded above and below by quadratic functions of $\xi$, and $S(t)$ is bounded and u-PE by Assumption~\ref{ass:PE}, the unforced system is uniformly exponentially stable (UES) on $[t_1,\tau)$ with a decay rate $\alpha_\xi > 0$ using standard contradiction arguments similar to \cite[Theorem 2.5.1]{SCC.Sastry.Bodson1989}. Because $A^{-1}$ and $Q^\top$ are uniformly bounded, the disturbance term satisfies $\|A^{-1}Q^\top d_f\|\leq \frac{1}{\alpha_1}\|Q\|_\infty\bar{d} \eqqcolon \bar{d}_\xi$. Using the exponential stability of the unforced system and the variation-of-constants formula, it follows that
\begin{equation}\label{eq:zeta_bound}
    \|\xi(t)\|
    \leq \iota_1 e^{-\alpha_\xi(t-t_1)}\|\xi(t_1)\|
    + \frac{\iota_1}{\alpha_\xi}\bar{d}_\xi
    < \infty,
\end{equation}
for all $t\in[t_1,\tau)$, where $\iota_1\coloneqq\sqrt{\overline{y}/\alpha_1}$. Since $\mathcal{E}_1(t)=A(t)\xi(t)$ and $A(t)$ is uniformly bounded, $\mathcal{E}_1(t)$ is also uniformly bounded. In particular,
\begin{equation}\label{eq:E1_bound}
    \|\mathcal{E}_1(t)\|
    \leq \overline{y}\iota_1 e^{-\alpha_\xi(t-t_1)}\|\xi(t_1)\|
    + \frac{\overline{y}\iota_1}{\alpha_\xi}\bar{d}_\xi
    < \infty.
\end{equation}

Projecting \eqref{eq:E_dot} onto $\phi^\perp$ yields the unexcited subspace dynamics
\begin{equation}\label{eq:E2_dot}
  \dot{\mathcal{E}}_2 = -\beta(t)\,\frac{C^\top(t)\,S(t)}{\operatorname{tr}(S(t)A(t))}\,\mathcal{E}_1.
\end{equation}
Note that there is no disturbance term in \eqref{eq:E2_dot} because $U_2^\top Y_f^\top = 0$. Furthermore, from the decomposition in Theorem~\ref{thm:Y_bounded}, the cross-coupling matrix $C$ evolves according to $\dot{C} = -\beta(t)\frac{A S}{\operatorname{tr}(SA)} C$. Defining $E_c(t) \coloneqq A^{-1}(t)C(t)$, identical algebraic steps to the $\mathcal{E}_1$ transformation show that $\dot{E}_c = -A^{-1}(t) S(t) E_c$. Therefore, $E_c$ is governed by the exact same UES unforced dynamics as $\xi$. Consequently, $\|E_c(t)\|\leq \iota_1 e^{-\alpha_\xi(t-t_1)}\|E_c(t_1)\|$, for all $t\in[t_1,\tau)$. Since $C(t)=A(t)E_c(t)$ and $A(t)$ is uniformly bounded,
\begin{equation}
    \|C(t)\|
    \leq \overline{y}\iota_1
    e^{-\alpha_\xi(t-t_1)}\|E_c(t_1)\|,
\end{equation}
and hence $C(t)$ converges to zero exponentially.

Since $\mathcal{E}_1$ is bounded by \eqref{eq:E1_bound}, $S/\operatorname{tr}(SA)$ is bounded (being defined as zero when $Q=0$), and $C^\top$ decays exponentially, the right-hand side of \eqref{eq:E2_dot} is bounded by a strictly exponentially decaying envelope as
\begin{equation}
  \|\dot{\mathcal{E}}_2(t)\| \le \iota_3\,e^{-\alpha_\xi(t-t_1)}, \quad\forall t\in [t_1,\tau),
\end{equation}
where $\iota_3 > 0$ is a finite constant dependent on the uniform bounds of $\mathcal{E}_1$ and $S$. Integrating this inequality yields
\begin{equation}
  \|\mathcal{E}_2(t)\| \le \|\mathcal{E}_2(t_1)\| + \frac{\iota_3}{\alpha_\xi} < \infty, \quad\forall t\in [t_1,\tau).
\end{equation}
Because both orthogonal components $\mathcal{E}_1$ and $\mathcal{E}_2$ are uniformly bounded, the total residual $\mathcal{E}(t) = U_1 \mathcal{E}_1(t) + U_2 \mathcal{E}_2(t)$ is uniformly bounded for all $t\in [t_1,\tau)$.
\end{proof}
\subsection{Sparse MRE Update Law}
\label{subsec:sparse_cost}

To promote sparse parameter estimates while maintaining consistency with the accumulated integral data, consider the regularized cost
\begin{equation}\label{eq:J}
    J(\hat{\theta},t)
    \coloneqq
    \frac{1}{2}\hat{\theta}^\top\mathcal{Y}(t)\hat{\theta}
    -\hat{\theta}^\top\mathcal{U}(t)
    +\lambda \|\hat{\theta}\|_1,
\end{equation}
where $\lambda>0$ is the sparsity penalty. The first two terms form a standard quadratic objective derived from the MRE structure. The additional $\ell_1$-regularization term $\lambda \|\hat{\theta}\|_1$ promotes sparsity in $\hat{\theta}$ by penalizing the magnitude of individual parameters, encouraging small or inactive entries to decay toward zero. The adaptive update law for $\hat{\theta}$ is derived to minimize the regularized cost \eqref{eq:J} using gradient descent. Because the $\ell_1$-regularization term is nonsmooth at the origin, a generalized gradient of $J$ is utilized, given by the Clarke subdifferential \cite[Definition 2.2]{SCC.Shevitz.Paden1994}
\begin{equation}\label{eq:partial_J}
    \partial_{\hat{\theta}}J(\hat{\theta},t)
    =
    \mathcal{Y}(t)\hat{\theta}-\mathcal{U}(t)+\lambda\operatorname{SGN}(\hat{\theta}).
\end{equation}
In \eqref{eq:partial_J}, the addition of a vector and a set is understood in the sense of the Minkowski sum, where we use the shorthand $a + A$ to denote $\{a\} + A$ for a vector $a \in \mathbb{R}^p$ and a set $A \subseteq \mathbb{R}^p$. Furthermore, $\operatorname{SGN}(\hat{\theta}) \subseteq \mathbb{R}^p$ denotes the set-valued sign map applied to the vector $\hat{\theta}$, defined as
\begin{equation}
    \operatorname{SGN}(\hat{\theta}) \coloneqq \left\{ \sigma \in \mathbb{R}^p \;\middle|\; \sigma_i \in \operatorname{SGN}(\hat{\theta}_i), \ \forall i = 1, \dots, p \right\},
\end{equation}
where $\operatorname{SGN}(\cdot)$ is defined for $z_i \in \mathbb{R}$ as
\begin{equation}\label{eq:SGN}
    \operatorname{SGN}(z_i)\coloneqq
    \begin{cases}
        \{1\}, & z_i>0,\\
        [-1,1], & z_i=0,\\
        \{-1\}, & z_i<0.
    \end{cases}
\end{equation}
%The first two terms represent the descent direction associated with the MRE, while the signum term biases the estimate toward sparse solutions.

The SP-VDF update law is given by the differential inclusion
\begin{equation}\label{eq:theta_hat_dot}
    \dot{\hat{\theta}} \in 
    \operatorname{Proj}_{\Theta_\varepsilon} \left(\hat{\theta}, \, \psi(t, \hat{\theta}), \, \Gamma \right) - k_{\theta}\lambda \Gamma \operatorname{SGN}(\hat{\theta}) \eqqcolon F(t, \hat{\theta}),
\end{equation}
where $\psi(t, \hat{\theta}) \coloneqq \Gamma Y(x(t))^\top e(t) + k_{\theta}\Gamma (\mathcal{U}(t) - \mathcal{Y}(t)\hat{\theta})$, $k_{\theta}>0$ is the MRE-learning gain, $\Gamma \in \mathbb{R}^{p \times p}$ is a diagonal, positive-definite adaptation gain matrix ($\Gamma = \Gamma^\top \succ 0$), and $\operatorname{Proj}_{\Theta_\varepsilon}(\cdot,\cdot, \cdot)$ denotes the smooth projection operator defined in \cite[Appendix E]{SCC.Krstic.Kanellakopoulos.ea1995}, employed to ensure that the parameter estimates remain bounded. To ensure continuity of the projection operator near the boundary of the parameter set, we introduce a boundary layer of thickness $\varepsilon>0$. 
%As shown in \cite[Appendix E]{SCC.Krstic.Kanellakopoulos.ea1995}, the resulting parameter estimates remain within the expanded set $\Theta_{\varepsilon} \coloneqq \left\{\hat{\theta}\in\mathbb{R}^{p} ,\middle|, \|\hat{\theta}\| \leq r_\theta +\varepsilon \right\}$.}

\section{Stability Analysis}\label{sec:stabilityAmalysis}

In this section, we establish ultimate boundedness of the closed-loop trajectories using nonsmooth Lyapunov analysis together with a contradiction argument.

Let $\zeta \coloneqq [e^\top, \hat{\theta}^\top, \operatorname{vec}(\mathcal{Y})^\top, \mathcal{U}^\top, \operatorname{vec}(Y_f)^\top, u_f^\top]^\top \in \mathbb{R}^{2n + 2p + p^2 + np}$ denote the augmented closed-loop state. The filtered signals in \eqref{eq:filtered_regression} are assumed to be generated by FIR filters for the rest of the analysis. In particular, the FIR filters in the MRE dynamics \eqref{eq:Y_dot}, \eqref{eq:U_dot}, and \eqref{eq:E_dot} depend on the history of $\zeta$. We denote this history by $\zeta_t$, where $\zeta_t: [-\Delta_t,0] \to \mathbb{R}^{2n + 2p + p^2 + np}$ is a continuous function defined as
\begin{equation}
    \zeta_t(s) \coloneqq \zeta(t+s), \quad s \in [-\Delta_t, 0],
\end{equation}
where $\Delta_t$ is the window length introduced in \eqref{eq:int_reg}. Thus, the closed-loop system can be represented as a functional differential inclusion (FDI). The dynamics consist of an accumulation phase and a directional-forgetting phase given by
\begin{equation}\label{eq:augmented_FDI}
    \dot{\zeta}(t) \in \begin{cases} 
    \mathcal{F}_1(t, \zeta_t), & t \in [t_0, t_1], \\ 
    \mathcal{F}_2(t, \zeta_t), & t \in (t_1, \infty), 
    \end{cases}
\end{equation}
where, for $i \in \{1,2\}$, the set-valued map $\mathcal{F}_i$ takes the form
\begin{equation}
    \mathcal{F}_i(t, \zeta_t) \coloneqq G_i(t, \zeta_t) + \begin{bmatrix} 0_{n \times 1} \\ -k_\theta \lambda \Gamma \operatorname{SGN}(\hat{\theta}(t)) \\ 0_{p^2 \times 1} \\ 0_{p \times 1} \\ 0_{np \times 1} \\ 0_{n \times 1} \end{bmatrix}, \quad i \in \{1, 2\},
\end{equation}
and $G_i$ collects the continuous parts of the dynamics, including the tracking-error \eqref{eq:closed_error}, the continuous projected component of \eqref{eq:theta_hat_dot}, the MRE dynamics \eqref{eq:Y_dot}--\eqref{eq:U_dot}, and the FIR definitions following \eqref{eq:int_reg}.
% with $G_i(t, \zeta_t)$ comprising the smooth portion of the closed-loop vector field, defined as
% \begin{equation}
%     G_i(t, \zeta_t) \coloneqq \begin{bmatrix}
%     -Ke(t) + \Delta(t,e(t)) \\
%     \operatorname{Proj}_{\Theta_\varepsilon} \left(\hat{\theta}(t), \, \psi(t, \hat{\theta}(t)), \, \Gamma \right) \\
%     \operatorname{vec}(\dot{\mathcal{Y}}_i(t)) \\
%     \dot{\mathcal{U}}_i(t) \\
%     \operatorname{vec}(\dot{Y}_f(t)) \\
%     \dot{u}_f(t)
%     \end{bmatrix},
% \end{equation}
% where $\dot{\mathcal{Y}}_i$ and $\dot{\mathcal{U}}_i$ correspond to the $i$-th phase of the MR dynamics defined in \eqref{eq:Y_dot} and \eqref{eq:U_dot}.

For each $i \in \{1,2\}$, $G_i$ is locally Lipschitz continuous and hence locally bounded, while $\operatorname{SGN}$ is upper semicontinuous with nonempty, compact, and convex values. Thus, $\mathcal{F}_i$ is an upper semicontinuous Carath\'eodory multifunction satisfying the upper semicontinuity and measurability conditions in (H1) and (H2) of \cite[Theorem~3.1]{SCC.Boudjenah.ea2010}. Since $\mathcal{F}_i$ is locally bounded on bounded sets of continuous history segments, the integrability condition (H3) is satisfied. In addition, the local Lipschitz continuity of $G_i$ and boundedness of $\operatorname{SGN}$ give the local growth bound in (H4). Hence, by \cite[Theorem~3.1]{SCC.Boudjenah.ea2010}, for every admissible continuous initial history $\zeta_{t_0}$ satisfying (H5), there exists an absolutely continuous solution $t \mapsto \zeta(t)$ of \eqref{eq:augmented_FDI} on the interval $[t_0,t_0+\delta)$, for some $\delta>0$. The same argument applies to both the accumulation and directional-forgetting parts of \eqref{eq:augmented_FDI}, and as a result, if a solution exists on $[t_0, t_1]$, then it can be continued to $[t_1,t_1+\delta)$ for some $\delta > 0$.

To facilitate the subsequent stability analysis, we present the following lemma on boundedness of $\hat{\theta}$.
\begin{lemma}\label{lem:theta_bound}
Suppose Assumptions~\ref{ass:bounds}--\ref{ass:PE} hold. If $[t_0, t_0+\delta)$ is the maximal interval of existence of a solution of \eqref{eq:augmented_FDI} such that $\hat{\theta}(t_{0}) \in \Theta_{\varepsilon}\coloneqq \left\{\hat{\theta}\in\mathbb{R}^{p} ,\middle|, \|\hat{\theta}\| \leq r_\theta +\varepsilon \right\}$, then $\Theta_{\varepsilon}$ is positively invariant under the differential inclusion \eqref{eq:theta_hat_dot}, i.e., $\hat{\theta}(t) \in \Theta_{\varepsilon}$ for all $t \in [t_0, t_0+\delta)$.
\end{lemma}
\begin{proof}
Let $\partial \Theta_{\varepsilon}$ and $T_{\Theta_{\varepsilon}}(\hat{\theta})$ denote the boundary and tangent cone of $\Theta_{\varepsilon}$, respectively. Since $\Theta_{\varepsilon}$ is a closed ball centered at the origin, the outward normal at any boundary point $\hat{\theta} \in \partial \Theta_{\varepsilon}$ is $n(\hat{\theta}) = \frac{\hat{\theta}}{\|\hat{\theta}\|}$. From \eqref{eq:theta_hat_dot}, it follows that for any $\rho \in F(t, \hat{\theta})$, there exists a $\sigma \in \operatorname{SGN}(\hat{\theta})$ such that
\begin{equation}
\rho = \operatorname{Proj}_{\Theta_\varepsilon}\Big(\hat{\theta}, \, \psi(t, \hat{\theta}), \, \Gamma\Big) - k_\theta \lambda \Gamma \sigma.
\end{equation}

By the properties of the smooth projection operator \cite[Lemma E.1]{SCC.Krstic.Kanellakopoulos.ea1995}, for $\hat{\theta} \in \partial \Theta_{\varepsilon}$, the projected vector field satisfies
\begin{equation}\label{eq:first_ineq}
n(\hat{\theta})^\top \operatorname{Proj}_{\Theta_\varepsilon}\Big(\hat{\theta}, \psi(t, \hat{\theta}), \Gamma\Big) \le 0.
\end{equation}

For the regularizing term $\nu \coloneqq -k_\theta \lambda \Gamma \sigma$, the positive definiteness and diagonal structure of $\Gamma$ ensure $\hat{\theta}^\top \Gamma \sigma = \sum_i \Gamma_{ii} |\hat{\theta}_i| \ge \lambda_{\min}(\Gamma) \|\hat{\theta}\|_1$. Since the origin lies in the interior of $\Theta_\epsilon$, $\nu$ points toward the origin. In particular, for $\hat{\theta} \in \partial \Theta_{\varepsilon}$ and $\sigma \in \operatorname{SGN}(\hat{\theta})$,
\begin{equation}\label{eq:second_ineq}
n(\hat{\theta})^\top \nu
= -(k_\theta \lambda) \left( \frac{\hat{\theta}}{\|\hat{\theta}\|} \right)^\top \Gamma \sigma \le -(k_\theta \lambda) \lambda_{\min}(\Gamma) \frac{\|\hat{\theta}\|_1}{\|\hat{\theta}\|} \le 0.
\end{equation}
Combining \eqref{eq:first_ineq} and \eqref{eq:second_ineq}, we obtain $n(\hat{\theta})^\top \rho \le 0$ for all $\rho \in F(t, \hat{\theta})$ on $\partial \Theta_{\varepsilon}$, implying $\rho \in T_{\Theta_{\varepsilon}}(\hat{\theta})$. By standard viability conditions, e.g., \cite[Theorem~7]{SCC.Aubin.Cellina1984b}, $\Theta_{\varepsilon}$ is positively invariant on the interval of existence $[t_0, t_0+\delta)$.
\end{proof}

%We next show that the solution is forward complete, i.e., $\delta=\infty$.
Since the true parameter satisfies $\theta \in \Theta$ by Assumption~\ref{ass:theta_bounded}, we conclude that $\|\tilde{\theta}(t)\| \le \|\hat{\theta}(t)\| + \|\theta\| \le 2r_\theta + \varepsilon$, and thus $\tilde{\theta}(t) \in \overline{B}(0,2r_\theta + \varepsilon)$ for all $t \in [t_0, t_0+\delta)$.

Next, we define $\Delta(t,e) \coloneqq Y\big(e + x_d(t)\big)\tilde{\theta}(t) + \epsilon\big(e + x_d(t)\big)$, so the error dynamics \eqref{eq:closed_error} can be written as
\begin{equation}\label{eq:error_field}
    \dot e = -Ke + \Delta(t,e).
\end{equation}
%Since Lemma \ref{lem:theta_bound} shows that $\tilde{\theta}$ remains bounded on the interval of existence of solutions of \eqref{eq:augmented_FDI} and  $\overline{x}_d \coloneqq \sup_{t\ge t_0}\|x_d(t)\| < \infty$ (Assumption \ref{ass:trajectory}), modeling the influence of $\tilde{\theta}$ on $e$ as a bounded disturbance term $\Delta$ on the interval of existence is justified. 
Choose $r_e>0$ such that $\overline{B}(0,r_e+\overline{x}_d)\subseteq \mathbb{X}$, where $\overline{x}_d \coloneqq \sup_{t\ge t_0}\|x_d(t)\| < \infty$ by Assumption \ref{ass:trajectory}. Then, if $e(t) \in \overline{B}(0,r_e)$, the system state satisfies $x(t) \in \overline{B}(0,r_e+\overline{x}_d) \subseteq \mathbb{X}$. Because the regressor matrix $Y$ is continuous and evaluated on a compact set, there exists a constant $\overline{Y}>0$ such that $\sup_{\|x\| \le r_e+\overline{x}_d} \|Y(x)\| \le \overline{Y}$. 
Since $\tilde{\theta}$ remains bounded on the interval of existence $[t_0,t_0+\delta)$ of solutions of \eqref{eq:augmented_FDI}, the disturbance term is bounded by
\begin{equation}\label{eq:Delta}
    \sup_{(t,e)\in[t_0,t_0+\delta)\times \overline{B}(0,r_e)}\|\Delta(t,e)\| \le \overline{\Delta} \coloneqq (2r_\theta + \varepsilon)\overline{Y} + \overline{\epsilon}.
\end{equation}
The following lemma establishes forward completeness of solutions of \eqref{eq:augmented_FDI}.
\begin{lemma}\label{lem:init_bounded}
Suppose Assumptions \ref{ass:bounds}--\ref{ass:PE} hold and $\hat{\theta}(t_0)\in\Theta_{\epsilon}$. Let $K=K^\top\succ 0$ and define $k \coloneqq \lambda_{\min}(K)$. If the gain conditions
\begin{equation}\label{eq:gain_condition_e}
\|e(t_0)\| < r_e \quad \text{and} \quad \frac{\overline{\Delta}}{k} < r_e
\end{equation}
are satisfied, then all maximal solutions of \eqref{eq:augmented_FDI} are forward complete, and the tracking error satisfies $e(t) \in \overline{B}(0,r_e)$ for all $t \ge t_0$. %
\end{lemma}
\begin{proof}
Let $\zeta(\cdot)$ denote a maximal solution of \eqref{eq:augmented_FDI}, starting from an initial condition that satisfies $\hat{\theta}(t_0)\in \Theta_\varepsilon$ and $e(t_0)\in B(0,r_e)$. Consider the candidate Lyapunov function $W:\mathbb{R}^n \to \mathbb{R}_{\ge 0}$ defined by
\begin{equation}\label{eq:W}
W(e) \coloneqq \frac{1}{2}\|e\|^2.
\end{equation}
Using \eqref{eq:error_field}, the time derivative of \eqref{eq:W} along the solution $\zeta(\cdot)$ exists for almost all $t \in [t_{0}, \, t_{0} + \delta)$ and satisfies
\begin{equation}
\dot{W}(t,e(t)) = -e^\top K e + e^\top \Delta(t,e) \le -k\|e(t)\|^2 + \|e(t)\|\,\|\Delta(t,e(t))\|.
\end{equation}
Young's inequality can be applied for almost all $t \in [t_{0}, \, t_{0} + \delta)$ where the solution is defined to yield
\begin{equation}
\dot{W}(t,e(t)) \le -k W(e(t)) + \frac{\|\Delta(t,e(t))\|^2}{2k}.
\end{equation}

In the following, we use a contradiction argument to show that the $e-$component of $\zeta$ cannot leave $\overline{B}(0,r_{e})$. For the sake of contradiction, assume that there exists some time $T \in (0, \delta)$ such that $e(t_{0} + T) \notin \overline{B}(0, r_{e})$.

Since $e(\cdot)$ is continuous and $e(t_0) \in B(0, r_{e})$, the Intermediate Value Theorem (IVT) \cite[Theorem 4.23]{SCC.Rudin1976} guarantees the existence of time instances $\rho_{1} \in (0,T)$ and $\rho_{2} \in (\rho_{1}, T]$ such that $e(t) \in \overline{B}(0, r_{e})$ for all $t \in [t_0, t_0+\rho_{1}]$ and $e(t) \notin \overline{B}(0, r_{e})$ for all $t \in (t_0+\rho_{1}, t_0+\rho_{2})$. Thus, for almost all $t \in [t_0, t_0+\rho_{1}]$, we have
\begin{equation}
    \dot{W}(t,e(t)) \le - k W(e(t)) + \frac{ \bar{\Delta}^2 }{2k}.
\end{equation}

Applying the Comparison Lemma (Lemma 3.4 in \cite{SCC.Khalil2002}), we get
\begin{equation}
W(e(t)) \le W(e(t_0))e^{-k(t-t_0)} + \frac{\bar\Delta^2}{2k^2}\big(1-e^{-k(t-t_0)}\big),
\end{equation}
which implies
\begin{equation}
\|e(t)\| \le \max\left\{\|e(t_0)\|,\frac{\bar\Delta}{k}\right\},
\quad \forall t \in [t_0, t_{0} + \rho_{1}].
\end{equation}

If the gain condition in \eqref{eq:gain_condition_e} is satisfied, then there exists a constant $\varsigma>0$ such that for all $t \in [t_0, t_{0} + \rho_{1}]$, $\|e(t)\| < r_e - \varsigma$, which, along with $e(t) \notin \overline{B}(0, r_{e})$ for all $t \in (t_0+\rho_{1}, t_0+\rho_{2})$ contradicts the continuity of $e$ at $t_0+\rho_1$. Therefore, our assumption that there exists $T \in (0, \delta)$ such that $e(t_0+T) \notin \overline{B}(0, r_{e})$ is false. We thus conclude that $e(t) \in \overline{B}(0, r_{e})$ for all $t \in [t_0, t_0+\delta)$. 

To establish forward completeness, first note that $e(\cdot)$ and $\hat{\theta}(\cdot)$ are uniformly bounded on $[t_0,t_0+\delta)$, with the latter following from Lemma~\ref{lem:theta_bound}. Since $x(t)=e(t)+x_d(t)$ and $x_d(\cdot)$ is bounded, the state $x(\cdot)$ is also bounded. The filtered signals $Y_f(\cdot)$, $u_f(\cdot)$, and $d_f(\cdot)$ are generated by finite integrals driven by bounded continuous functions of $x(\cdot)$, and are therefore bounded on $[t_0, t_0+\delta)$. Because $Y_f(\cdot)$ and $d_f(\cdot)$ are  bounded, Theorems~\ref{thm:Y_bounded} and \ref{thm:E_bounded} imply that the memory regressor $\mathcal{Y}(\cdot)$ and the recursive residual $\mathcal{E}(\cdot)$ are bounded on $[t_0,t_0+\delta)$. Furthermore, since $\mathcal{U}(t) = \mathcal{Y}(t)\theta + \mathcal{E}(t)$ for all $t\in[t_0,t_0+\delta)$, it follows that $\mathcal{U}(\cdot)$ is also bounded on $[t_0,t_0+\delta)$. Furthermore, since the bounds on $\hat{\theta}$ and $e$ are independent of $\delta$, the solution $t \mapsto \zeta(t)$ is bounded on $[t_0,t_0+\delta)$, with a bound that is independent of $\delta$. The solution is thus precompact.
% A solution being precompact means the set \overline{\zeta(t)\mid t\in I} is compact on every interval of existence I.
Since precompact maximal solutions are forward complete, we conclude that the maximal interval of existence is $[t_0, \infty)$. Since $\zeta$ was arbitrary, the result holds strongly, i.e., for all solutions of \eqref{eq:augmented_FDI} that satisfy $\hat{\theta}(t_0)\in \Theta_\varepsilon$ and $e(t_0)\in B(0,r_e)$.
\end{proof}

We now establish the ultimate bound on the error state $z \coloneqq [e^\top,\tilde{\theta}^\top]^\top\in\mathbb{R}^{n+p}$. In the theorem below, the treatment of all other components of $\zeta$ as bounded time-varying signals is justified by their boundedness established in Lemma \ref{lem:init_bounded}.
\begin{theorem}\label{thm:main}
    Under the hypothesis of Lemma \ref{lem:init_bounded}, any solution of \eqref{eq:augmented_FDI} such that $\hat{\theta}(t_0)\in \Theta_\varepsilon$ and $e(t_0)\in B(0,r_e)$ satisfies 
    \begin{equation}\label{eq:UUBound}
        \limsup_{t \to \infty} \|z(t)\| \le 
        \underline{v}^{-1}\!\left(\overline{v}\left(\sqrt{\iota/\kappa}\right)\right),
    \end{equation}
    where $\overline{v}$, $\underline{v}$ are class $\mathcal{K}_{\infty}$ functions defined in \eqref{eq:Vbounds}, $\kappa \coloneqq \min\{\frac{k}{2}, \frac{k_{\theta}\underline{y}}{2}\}$, $\iota\coloneqq \frac{\overline{\epsilon}^2}{2k} +\frac{k_{\theta}(\overline{\mathcal{E}}+\lambda\sqrt{p})^2}{2\underline{y}}$, with $\underline{y} > 0$ being the lower bound such that $\mathcal{Y}(t) \succeq \underline{y} I_p$ for all $t \ge t_1$ established in Theorem \ref{thm:Y_bounded}, and $\overline{\mathcal{E}} \coloneqq \sup_{t \in \mathbb{R}_{\ge 0}} \|\mathcal{E}(t)\| < \infty$.
\end{theorem}
\begin{proof}
    Theorem \ref{thm:Y_bounded} implies that there exists $\underline{y}>0$ such that $\mathcal{Y}(t) \succeq \underline{y} I_p$ for all $t \ge t_1$. Consider the candidate Lyapunov function $V:\mathbb{R}^{n+p}\to\mathbb{R}$ defined by
    \begin{equation}\label{eq:V}
    V(z)\coloneqq \frac{1}{2}\|e\|^2+\frac{1}{2}\tilde{\theta}^{\top}\Gamma^{-1}\tilde{\theta},
    \end{equation}
    which satisfies the bounds
    \begin{equation}\label{eq:Vbounds}
    \underline{v}(\|z\|)\le V(z)\le \overline{v}(\|z\|), \qquad \forall z\in\mathbb{R}^{n+p},
    \end{equation}
    where $\underline{v}(r)\coloneqq \frac{1}{2}\min\!\left\{1,\lambda_{\min}(\Gamma^{-1})\right\}r^2$ and $\overline{v}(r)\coloneqq \frac{1}{2}\max\!\left\{1,\lambda_{\max}(\Gamma^{-1})\right\}r^2$ are class $\mathcal K_{\infty}$ functions. By the chain rule for differential inclusions \cite[Theorem~2.2]{SCC.Shevitz.Paden1994}, the function $t \mapsto V(z(t))$ is absolutely continuous, ensuring its time derivative exists almost everywhere. Along the system trajectories for $t \ge t_1$, using the property of projection operators in \cite[Lemma E.1. IV]{SCC.Krstic.Kanellakopoulos.ea1995}, the derivative satisfies
    \begin{equation}\label{eq:Vdot2}
    \dot V(t,z(t)) \le -k\|e(t)\|^2 -k_{\theta}\,\tilde\theta^\top(t)\mathcal{Y}(t)\tilde\theta(t) +\overline{\epsilon}\|e(t)\| -k_{\theta}\,\tilde\theta^\top(t)\bigl(\mathcal{E}(t)-\lambda\sigma(t)\bigr),
    \end{equation}
    for some measurable selection $\sigma(t) \in \operatorname{SGN}(\hat{\theta}(t))$. Applying the bounds $\mathcal{Y}(t) \succeq \underline{y} I_p$ and $\|\mathcal{E}(t)\| \le \overline{\mathcal{E}}$ valid for all $t \ge t_1$, and using the vector property $\|\sigma(t)\| \le \sqrt{p}$, Young's inequality yields $\dot V(t,z(t)) \le -\kappa\|z(t)\|^2 + \iota$, i.e., from \eqref{eq:Vbounds},
    \begin{equation}\label{eq:Vdot_main}
        \dot V(t,z(t)) \le -\frac{2\kappa}{ \max\{1,\lambda_{\max}(\Gamma^{-1})\} } V(z(t)) + \iota, \quad \text{a.e. } t \ge t_1,
    \end{equation}
    where $\kappa$ and $\iota$ are positive constants defined below \eqref{eq:UUBound}. Using the comparison lemma,
    \begin{multline}
    V(z(t)) \leq V(z(t_1)) \exp\left( -\frac{2\kappa}{ \max\{1,\lambda_{\max}(\Gamma^{-1})\} } (t - t_1) \right)\\ + \overline{v}\!\Bigl(\sqrt{\iota/\kappa}\Bigr) \left[ 1 - \exp\left( -\frac{2\kappa}{ \max\{1,\lambda_{\max}(\Gamma^{-1})\} } (t - t_1) \right) \right].
    \end{multline}
    Taking the limit supremum as $t \to \infty$ yields
    \begin{equation}
    \limsup_{t \to \infty} V(z(t)) \le \overline{v}\!\Bigl(\sqrt{\iota/\kappa}\Bigr).
    \end{equation}
    Using the lower bound from \eqref{eq:Vbounds} and the strictly increasing inverse $\underline{v}^{-1}$, we obtain
    \begin{equation}
    \limsup_{t \to \infty} \|z(t)\| \le \underline{v}^{-1}\Bigl(\limsup_{t \to \infty} V(z(t))\Bigr) \le \underline{v}^{-1}\!\Bigl(\overline{v}\bigl(\sqrt{\iota/\kappa}\bigr)\Bigr),
    \end{equation}
    which establishes the ultimate bound in \eqref{eq:UUBound}.
\end{proof}

\section{Simulation Study}
\label{sec:simulation}

This section evaluates the SP-VDF update law in \eqref{eq:theta_hat_dot} in terms of closed-loop tracking, parameter estimation, and online support recovery (identifying the nonzero entries of the unknown parameter vector) in the presence of function-approximation errors and parameter changes. Since the stability analysis assumes a constant unknown parameter vector, parameter switching is used only to evaluate alertness empirically and lies outside the hypotheses of Theorem~\ref{thm:main}.

We consider a two-state control-affine system on a compact domain with dynamics given by \eqref{eq:dynamics}, where $x=[x_1,x_2]^\top\in\mathbb{R}^2$. The unknown drift dynamics are modeled as a modified Van der Pol oscillator,
\begin{equation}\label{eq:sim_dyn}
f(x,t)=
\begin{bmatrix}
x_2\\
-x_1+\mu(t)(1-x_1^2)x_2
\end{bmatrix},
\end{equation}
with $g(x)=I_2$. The damping coefficient is switched at $t=t_f/2=250\,\mathrm{s}$ according to $\mu(t)=1.0+0.5\,\mathbb{I}_{\{t\geq250\,\mathrm{s}\}}$, where $\mathbb{I}_{\{\cdot\}}$ denotes the indicator function and $t_f=500\,\mathrm{s}$ is the simulation horizon. Thus, $\mu(t)$ changes from $1.0$ to $1.5$ at $t=250\,\mathrm{s}$.

To evaluate the SP-VDF formulation in an overparameterized setting, the unknown vector field $f(x,t)$ is represented using a high-dimensional candidate dictionary as in \eqref{eq:f_approx}. The library matrix is constructed as $Y(x) = \text{diag}(\Phi(x)^\top, \Phi(x)^\top) \in \mathbb{R}^{2\times28}$, where the basis vector $\Phi(x)\in\mathbb{R}^{14}$ contains monomials up to degree three, trigonometric terms, and state cross-products
\begin{align}
\Phi(x) = \big[ 1, x_1, x_2, x_1^2, x_1x_2, x_2^2, x_1^3, x_1^2x_2, x_1x_2^2, x_2^3, \sin(x_1), \sin(x_2), x_1\sin(x_2), x_2\sin(x_1) \big]^\top.
\end{align}
To assess the effect of function approximation errors, $\epsilon(x)=0.02[\sin(2x_1),\cos(2x_2)]^\top$ is added to the dynamics as defined in \eqref{eq:f_approx}. In the resulting 28-dimensional parameter vector $\theta(t)\in\mathbb{R}^{28}$, only four entries are nonzero. Specifically, $\theta_{3} = 1$, $\theta_{16} = -1$, $\theta_{17}(t) = \mu(t)$, and $\theta_{22}(t) = -\mu(t)$, while the remaining $24$ elements of $\theta(t)$ are identically zero.

The system is initialized with the state $x(0)=[-2,1]^\top$ and without any prior knowledge of the unknown parameters, with the parameter estimates initialized as $\hat{\theta}(0)=0_{28 \times 1}$. The MRE signals are initialized as $\mathcal{Y}(0)=0_{28\times28}$ and $\mathcal{U}(0)=0_{28 \times 1}$.
% To examine parameter estimation under FE followed by subspace PE, the desired trajectory is chosen as
% \begin{multline}
% x_d(t) =
% \begin{bmatrix}
% 2.5\sin(0.7t) \\
% 2.0\cos(1.1t)
% \end{bmatrix}
% \\+ e^{-0.05t}
% \begin{bmatrix}
% 1.5\cos(2.3t) + 0.5\sin(4.1t)\\
% 1.2\sin(2.9t) + 0.4\cos(4.7t)
% \end{bmatrix},
% \end{multline}
% to satisfy Assumption~\ref{ass:FE} and \ref{ass:PE}. The exponentially decaying terms provide additional excitation during the initial portion of the simulation and gradually vanish, leaving the lower-frequency components before the parameter switch. 
Motivated by Assumptions~\ref{ass:FE} and \ref{ass:PE}, the desired trajectory is selected as
\begingroup\medmuskip=0mu\begin{equation}\label{eq:xd}
x_d(t)=
\begin{bmatrix}
2.5\sin(0.7t)\\
2.0\cos(1.1t)
\end{bmatrix}
+\eta(t)
\begin{bmatrix}
1.5\cos(2.3t)+0.5\sin(4.1t)\\
1.2\sin(2.9t)+0.4\cos(4.7t)
\end{bmatrix},
\end{equation}
where
\begin{equation}\label{eq:excitation_envelope}
\eta(t)=
\begin{cases}
\left(1-\dfrac{t}{t_e}\right)^3, & 0\leq t\leq t_e,\\[1mm]
0, & t>t_e,
\end{cases}
\end{equation}\endgroup
with the excitation cutoff time $t_e=100\,\mathrm{s}$. The filtered regressor $Y_f$ is generated using an FIR window of length $\Delta t=0.25\,\mathrm{s}$. Hence, the effect of the excitation cutoff persists in $Y_f(t)$ for at most $\Delta t$, so that for $t\geq t_1\coloneqq t_e+\Delta t=100.25\,\mathrm{s}$, $Y_f(t)$ contains no contribution from the additional excitation terms. The FE condition is verified numerically over the specified finite interval. Since subspace PE requires the integral condition to hold for all $t$, it is not verified directly here. Figure~\ref{fig:vdp_regressor_bound_min_all_lambda} shows the uniform positivity of the minimum eigenvalue of the MR throughout the simulation. The feedback control gain in \eqref{eq:control} is selected as $K=10I_{2}$, the scalar MRE gain in \eqref{eq:theta_hat_dot} is selected as $k_{\theta}=1$, and the adaptation gain matrix is $\Gamma=\operatorname{diag}(\gamma_1,\ldots,\gamma_{28})$, with $\gamma_{17}=\gamma_{22}=10$ and $\gamma_i=1$ for all other indices. The maximum forgetting rate for \eqref{eq:beta_composite} is selected as $\beta_{\max} = 5$ and the alertness thresholds are chosen as $\underline{y}_1=10^{-4}$ and $\underline{y}_2=0.20$, while the upper bounds are $\overline{y}_1=1$ and $\overline{y}_2=3$. Simulations are run using MATLAB's \texttt{dde23} solver with relative and absolute tolerances of $10^{-3}$ and $10^{-4}$, respectively.

\subsection{Comparative Baseline Formulations}

To assess the performance of directional forgetting, the proposed SP-VDF method is compared with two sparse parameter estimation baselines under identical initial conditions and reference trajectories:
\begin{enumerate}
    \item \textit{Sparsity-Promoting Integral Concurrent Learning (SP-ICL):} Based on \cite{SCC.Satharasi.Ogri.ea2026a}, this method uses a history stack of previously recorded data. The stack stores up to $N=100$ samples and is updated using a minimum recording interval of $t_{\min}=0.05\,\mathrm{s}$ and a minimum singular value threshold of $\lambda^*=10^{-3}$.
    
    \item \textit{Sparsity-Promoting Variable-Rate Uniform Forgetting (SP-VUF):} Based on the variable-rate uniform forgetting (VUF) estimator in \cite{SCC.Pan.Aranovskiy.ea2019}, this formulation includes the same sparsity-promoting $\ell_1$ regularization term used in SP-VDF. Unlike directional forgetting, the VUF approach uses a scalar forgetting factor that uniformly discounts the MR across all parameter directions. The forgetting parameters are chosen to match those of SP-VDF: $\beta_{\max}=5$, $\underline{y}_1=10^{-4}$, and $\underline{y}_2=0.20$.
\end{enumerate}

To examine the effect of the sparsity penalty, simulations are conducted for $\lambda \in \{0,10^{-3},5\times10^{-3},10^{-2},5\times10^{-2}\}$ for all three approaches. The case $\lambda=0$ corresponds to removing the sparsity penalty. Performance is evaluated using the root-mean-square (RMS) tracking error, terminal parameter estimation error $\|\tilde{\theta}(t_f)\|$, and support recovery $F_1$ score \cite[Chapter~8]{SCC.Manning.Raghavan2009}. A parameter is classified as active when $|\hat{\theta}_i| > \vartheta$, with $\vartheta \geq 0$. Support recovery is evaluated for $\vartheta \in \{1.0\times10^{-3},1.0\times10^{-2},5.0\times10^{-2},1.0\times10^{-1}\}$, as shown in Table~\ref{tab:ablation_tolerances}. The threshold distinguishes nonzero parameters from small estimates caused by numerical integration errors.

\begin{figure}
\centering
\resizebox{0.75\columnwidth}{!}{$
\input{figures/parameter_errors}
$}
\caption{Parameter estimation error norm $\|\tilde{\theta}(t)\|$ of the SP-VDF update law for different values of the regularization parameter $\lambda$.}
\label{fig:vdp_parameter_error_norm_all_lambda}
\end{figure}
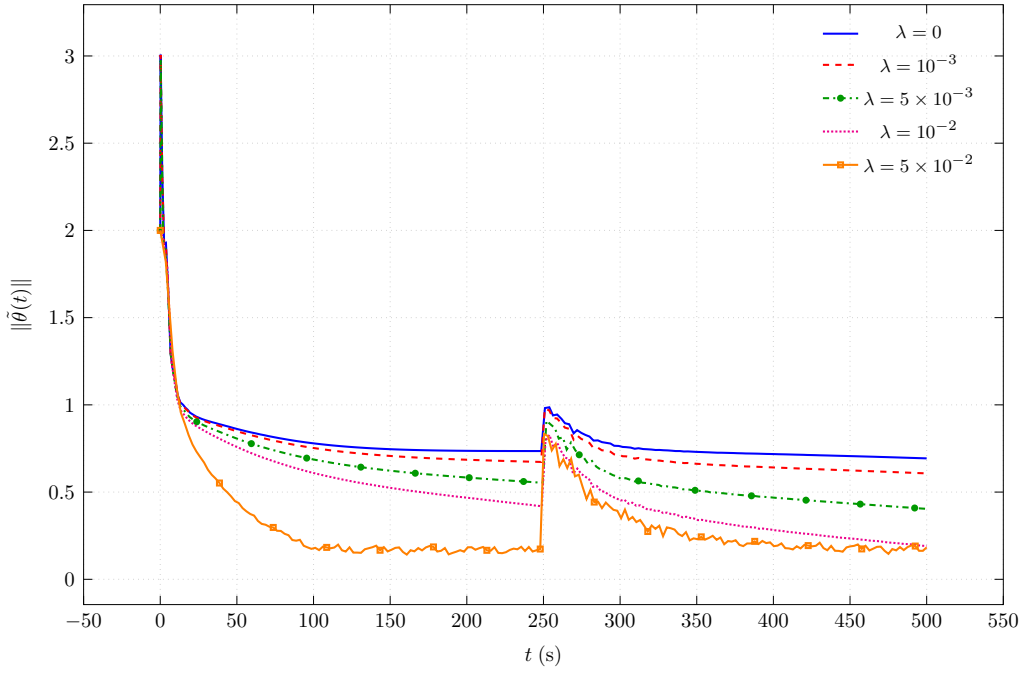

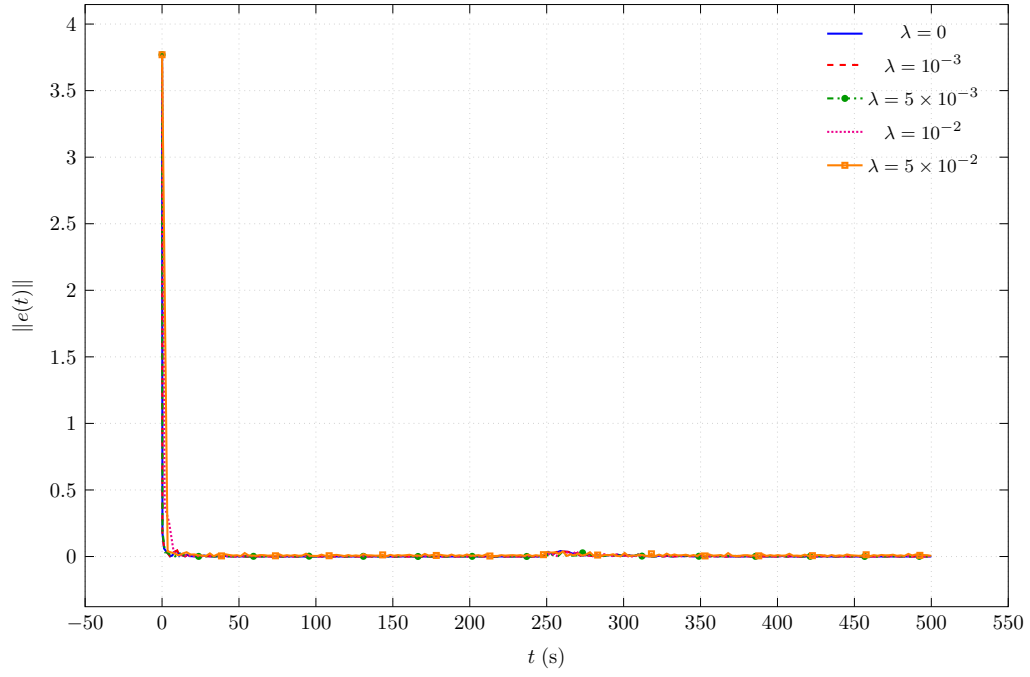
\begin{figure}
\centering
\resizebox{0.75\columnwidth}{!}{$
\input{figures/tracking_errors}
$}
\caption{Tracking error norm $\|e(t)\|$ of the SP-VDF update law for different values of the regularization parameter $\lambda$.}
\label{fig:vdp_tracking_error_norm_all_lambda}
\end{figure}

\begin{figure}
\centering
\resizebox{0.75\columnwidth}{!}{$
\input{figures/regressor_matrix_bounds}
$}
\caption{Minimum eigenvalue of $\mathcal{Y}(t)$ for the SP-VDF update law for different values of the regularization parameter $\lambda$.}
\label{fig:vdp_regressor_bound_min_all_lambda}
\end{figure}
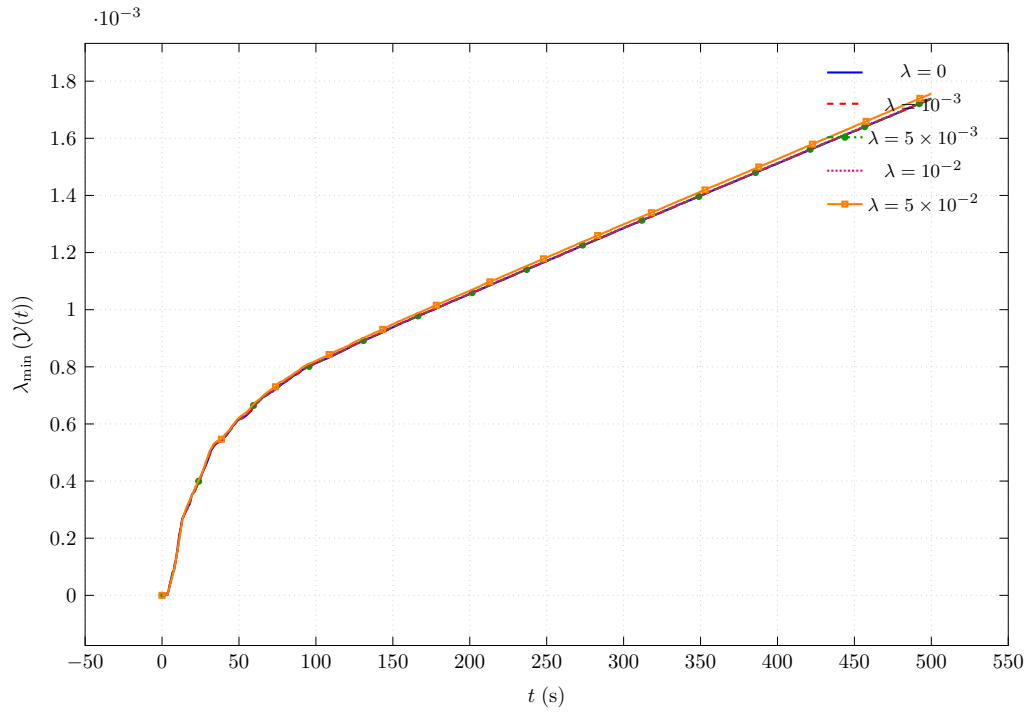

\begin{figure}
\centering
\resizebox{0.75\columnwidth}{!}{$
\input{figures/regressor_matrix_bounds2}
$}
\caption{Maximum eigenvalue of $\mathcal{Y}(t)$ for the SP-VDF update law for different values of the regularization parameter $\lambda$.}
\label{fig:vdp_regressor_bound_max_all_lambda}
\end{figure}
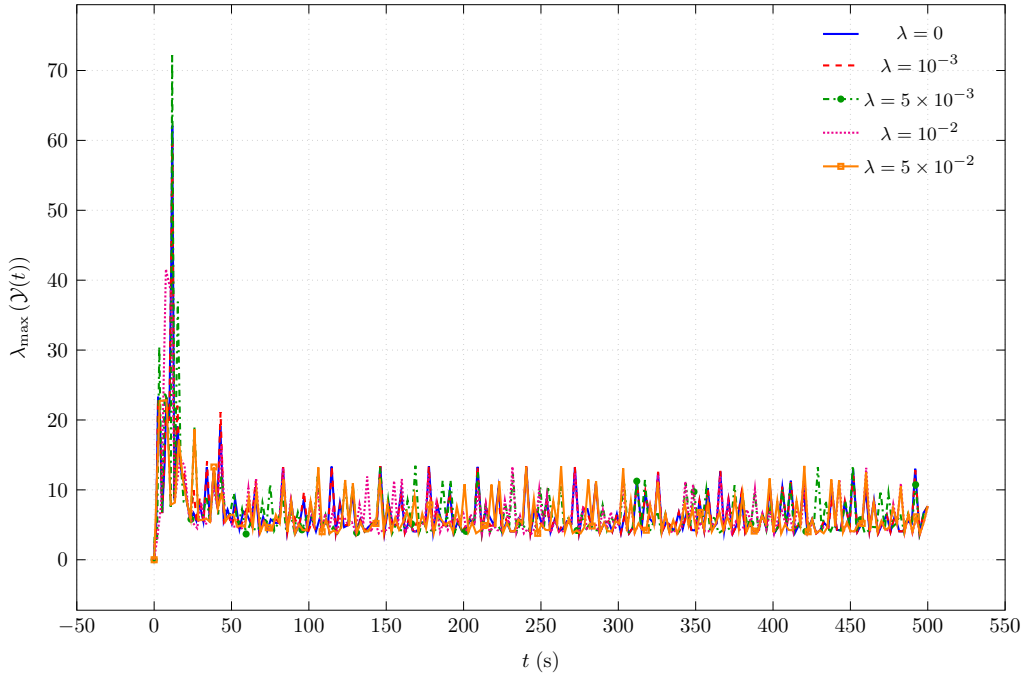

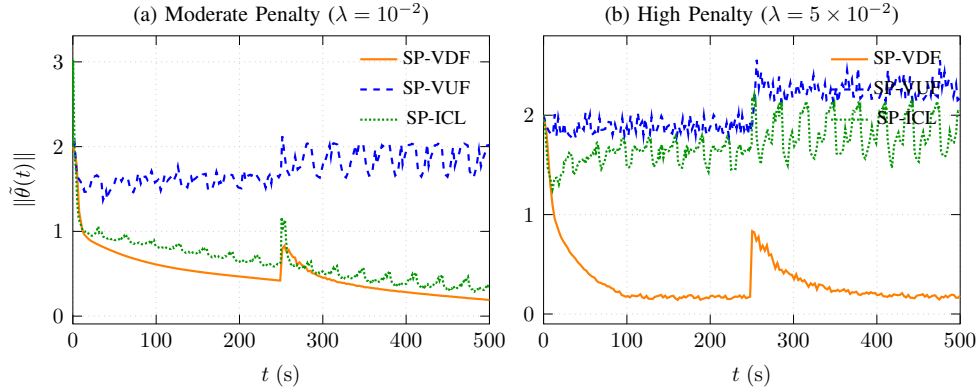
\begin{figure*}[t]
\centering
\resizebox{0.75\columnwidth}{!}{$
\input{figures/comparison_parameter_errors}
$}
\caption{Parameter estimation error norm $\|\tilde{\theta}(t)\|$ for SP-VDF, SP-VUF, and SP-ICL. (a) For $\lambda=10^{-2}$, SP-VDF recovers the parameters faster after the change at $t=250\,\mathrm{s}$ and reaches a smaller error sooner than the baseline methods. (b) For $\lambda=5\times10^{-2}$, SP-VDF remains stable and reduces the error to $0.1831$, while SP-VUF exhibits estimator windup and SP-ICL has a larger terminal error and loses numerical rank.}
\label{fig:error_comparison}
\end{figure*}

\begin{table}[ht]
\centering
\caption{Tracking and parameter estimation performance for different sparsity weights $\lambda$ ($t_f=500\,\mathrm{s}$). Best values are shown in bold.}
\label{tab:performance_lambda_sweep_all}
\vspace{0.2cm}
\renewcommand{\arraystretch}{1.2}
\begin{tabular}{@{} l ccc ccc @{}}
\toprule
& \multicolumn{3}{c}{\textbf{RMS} $\|e\|$ ($\times 10^{-2}$)} & \multicolumn{3}{c}{\textbf{Final} $\|\tilde{\theta}(t_f)\|$} \\
\cmidrule(lr){2-4} \cmidrule(lr){5-7}
$\lambda$ & SP-VDF & SP-VUF & SP-ICL & SP-VDF & SP-VUF & SP-ICL \\
\midrule
$0$                & 3.7927 & 3.7830 & \textbf{3.6635} & \textbf{0.6933} & 0.9830 & 0.9425 \\
$10^{-3}$          & 3.7888 & 3.8077 & \textbf{3.6604} & \textbf{0.6072} & 0.9157 & 0.8443 \\
$5\times10^{-3}$   & 3.7756 & 4.3324 & \textbf{3.6703} & \textbf{0.4040} & 1.1901 & 0.5164 \\
$10^{-2}$          & 3.7641 & 5.0981 & \textbf{3.7507} & \textbf{0.1911} & 2.0196 & 0.3861 \\
$5\times10^{-2}$   & \textbf{3.8394} & 6.6402 & 5.7252 & \textbf{0.1831} & 2.1735 & 2.0435 \\
\bottomrule
\end{tabular}
\end{table}

\begin{table}[ht]
\centering
\caption{$F_1$ scores at the parameter switch ($t=250\,\mathrm{s}$) and final time ($t_f=500\,\mathrm{s}$) for $\vartheta=10^{-3}$. Best values are shown in bold.}
\label{tab:f1_combined}
\vspace{0.2cm}
\renewcommand{\arraystretch}{1.2}
\begin{tabular}{@{} l cc cc cc @{}}
\toprule
& \multicolumn{2}{c}{\textbf{SP-VDF}} & \multicolumn{2}{c}{\textbf{SP-VUF}} & \multicolumn{2}{c}{\textbf{SP-ICL}} \\
\cmidrule(lr){2-3} \cmidrule(lr){4-5} \cmidrule(lr){6-7}
$\lambda$ & $t=250\,\mathrm{s}$ & $t=500\,\mathrm{s}$ & $t=250\,\mathrm{s}$ & $t=500\,\mathrm{s}$ & $t=250\,\mathrm{s}$ & $t=500\,\mathrm{s}$ \\
\midrule
$0$                & \textbf{0.3333} & \textbf{0.3478} & 0.2581 & 0.2759 & 0.2963 & 0.2857 \\
$10^{-3}$          & \textbf{0.4000} & 0.3478          & 0.3810 & 0.3478 & \textbf{0.4000} & \textbf{0.3636} \\
$5\times10^{-3}$   & \textbf{0.5000} & \textbf{0.5333} & 0.3810 & 0.2963 & 0.4211 & 0.3810 \\
$10^{-2}$          & \textbf{0.6667} & \textbf{0.5714} & 0.3158 & 0.2500 & 0.4211 & 0.3636 \\
$5\times10^{-2}$   & \textbf{0.8889} & \textbf{0.7273} & 0.2222 & 0.1429 & 0.4000 & 0.3000 \\
\bottomrule
\end{tabular}%
\end{table}

\begin{table}[ht]
\centering
\caption{Support recovery for different thresholds $\vartheta$ with $\lambda=0.05$ at $t_f=500\,\mathrm{s}$. Est., TP, and FP denote the estimated active count, true positives, and false positives, respectively. Best values are shown in bold.}
\label{tab:ablation_tolerances}
\vspace{0.2cm}
\renewcommand{\arraystretch}{1.2}
\begin{tabular}{@{} l cccc cccc cccc @{}}
\toprule
& \multicolumn{4}{c}{\textbf{SP-VDF}} & \multicolumn{4}{c}{\textbf{SP-VUF}} & \multicolumn{4}{c}{\textbf{SP-ICL}} \\
\cmidrule(lr){2-5} \cmidrule(lr){6-9} \cmidrule(lr){10-13}
$\vartheta$ & Est. & TP & FP & $F_1$ & Est. & TP & FP & $F_1$ & Est. & TP & FP & $F_1$ \\
\midrule
$1.0\times10^{-3}$ & 7 & 4 & 3 & \textbf{0.7273} & 10 & 1 & 9 & 0.1429 & 16 & 3 & 13 & 0.3000 \\
$1.0\times10^{-2}$ & 4 & 4 & 0 & \textbf{1.0000} & 8  & 1 & 7 & 0.1667 & 15 & 3 & 12 & 0.3158 \\
$5.0\times10^{-2}$ & 4 & 4 & 0 & \textbf{1.0000} & 3  & 1 & 2 & 0.2857 & 8  & 3 & 5  & 0.5000 \\
$1.0\times10^{-1}$ & 4 & 4 & 0 & \textbf{1.0000} & 3  & 1 & 2 & 0.2857 & 5  & 3 & 2  & 0.6667 \\
\bottomrule
\end{tabular}%
\end{table}

\begin{table}[ht]
\centering
\caption{MR conditioning and effective rank at the parameter switch ($t=250\,\mathrm{s}$) and final time ($t_f=500\,\mathrm{s}$) for $\lambda=5\times10^{-2}$. Effective rank is defined as the number of eigenvalues greater than a numerical tolerance of $10^{-6}$. Best values are shown in bold.}
\label{tab:matrix_conditioning}
\vspace{0.2cm}
\renewcommand{\arraystretch}{1.2}
\begin{tabular}{@{} l cc cc cc @{}}
\toprule
& \multicolumn{2}{c}{\textbf{SP-VDF}} & \multicolumn{2}{c}{\textbf{SP-VUF}} & \multicolumn{2}{c}{\textbf{SP-ICL}} \\
\cmidrule(lr){2-3} \cmidrule(lr){4-5} \cmidrule(lr){6-7}
Time & Eff. Rank & $\lambda_{\min}(\mathcal{Y})$ & Eff. Rank & $\lambda_{\min}(\mathcal{Y})$ & Eff. Rank & $\lambda_{\min}(\mathcal{Y})$ \\
\midrule
$t=250\,\mathrm{s}$ & \textbf{28/28} & $\mathbf{1.18\times10^{-3}}$ & 10/28 & $6.51\times10^{-12}$ & 26/28 & $6.85\times10^{-8}$ \\
$t=500\,\mathrm{s}$ & \textbf{28/28} & $\mathbf{1.76\times10^{-3}}$ & 16/28 & $8.81\times10^{-12}$ & 24/28 & $4.88\times10^{-8}$ \\
\bottomrule
\end{tabular}%
\end{table}

\subsection{Discussion}

Figures~\ref{fig:vdp_parameter_error_norm_all_lambda}--\ref{fig:error_comparison} and Tables~\ref{tab:performance_lambda_sweep_all}--\ref{tab:matrix_conditioning} compare the effects of sparsity regularization and forgetting on closed-loop tracking, parameter estimation, MR conditioning, and support recovery.

\subsubsection{Closed-Loop Tracking Performance}

Figure~\ref{fig:vdp_tracking_error_norm_all_lambda} and Table~\ref{tab:performance_lambda_sweep_all} show that all three methods track the reference trajectory throughout the simulation. For $\lambda\leq10^{-2}$, SP-ICL achieves lower RMS tracking errors than SP-VDF. The transient response following the parameter switch provides new data that allows SP-ICL to update the changed parameter and reduce the tracking error, despite the rank deficiency of its MR shown in Table~\ref{tab:matrix_conditioning}.

At $\lambda=5\times10^{-2}$, the RMS tracking errors increase to $5.7252\times10^{-2}$ for SP-ICL and $6.6402\times10^{-2}$ for SP-VUF, compared with $3.8394\times10^{-2}$ for SP-VDF. The increase in error for SP-ICL and SP-VUF is consistent with the loss of excitation in several parameter directions after the finite excitation phase ends at $t=100\,\mathrm{s}$. The $\ell_1$ penalty then has a larger effect on the weakly excited directions. SP-VDF maintains a full-rank MR, as shown in Table~\ref{tab:matrix_conditioning}, and has the lowest RMS tracking error at this value of $\lambda$.

\subsubsection{Parameter Convergence and Support Recovery}

Figure~\ref{fig:vdp_parameter_error_norm_all_lambda} shows the parameter estimation error for each value of $\lambda$. For $\lambda\leq10^{-2}$, SP-VDF shows an increase in estimation error during the reduced-excitation interval. Directional forgetting preserves the eigenvalues associated with unexcited directions and prevents the MR from losing rank. Without new information in these directions, tracking and numerical errors can accumulate in the corresponding parameter estimates.

For $\lambda=5\times10^{-2}$, SP-VDF achieves a terminal estimation error of $\|\tilde{\theta}(t_f)\|=0.1831$, compared with $2.0435$ for SP-ICL and $2.1735$ for SP-VUF. Figure~\ref{fig:error_comparison}(b) shows the degradation of the baseline methods. Figures~\ref{fig:vdp_regressor_bound_min_all_lambda} and \ref{fig:vdp_regressor_bound_max_all_lambda}, together with Table~\ref{tab:matrix_conditioning}, show the corresponding changes in MR conditioning. Under directional forgetting, $\lambda_{\min}(\mathcal{Y})$ remains positive, and the effective rank is $28/28$ at both $t=250\,\mathrm{s}$ and $t=500\,\mathrm{s}$.

At $t=250\,\mathrm{s}$, SP-VUF has an effective rank of $10/28$ with $\lambda_{\min}(\mathcal{Y})=6.51\times10^{-12}$. The uniform forgetting update becomes sensitive to the resulting ill-conditioning, leading to estimator windup and increased solver execution times. SP-ICL retains an effective rank of $26/28$ with $\lambda_{\min}(\mathcal{Y})\approx6.85\times10^{-8}$. The loss of effective rank in both baseline methods occurs after the rich frequency components in the reference trajectory terminate at $t=100\,\mathrm{s}$.

At the parameter switch, SP-VDF recovers the correct support while adapting the changed coefficient values, achieving an $F_1$ score of $0.8889$ for $\lambda=5\times10^{-2}$. SP-ICL and SP-VUF achieve $F_1$ scores of $0.4000$ and $0.2222$, respectively.

\subsubsection{Sensitivity to Identification Thresholds}

Table~\ref{tab:f1_combined} reports the $F_1$ scores over the $\lambda$ sweep. The scores are computed using $\vartheta=10^{-3}$. For $\lambda=5\times10^{-2}$, SP-VDF has an $F_1$ score of $0.7273$ at $t_f=500\,\mathrm{s}$.

Table~\ref{tab:ablation_tolerances} reports support recovery for $\lambda=5\times10^{-2}$ at different thresholds. At $\vartheta=10^{-2}$, SP-VDF identifies the four nonzero parameters with zero false positives and an $F_1$ score of $1.0000$. The same $F_1$ score is obtained for $\vartheta=5\times10^{-2}$ and $\vartheta=10^{-1}$. For SP-VUF and SP-ICL, the estimated support differs from the true support over the tested threshold range. At $\vartheta=10^{-1}$, SP-VUF identifies three true active coefficients with two false positives, while SP-ICL identifies three true active coefficients with two false positives.

\section{Conclusion}
\label{sec:conclusion}

An online sparse identification and adaptive tracking framework for nonlinear control-affine systems with overparameterized candidate dictionaries is developed using a sparse MRE. The proposed SP-VDF update law combines recursive least-squares estimation with an $\ell_1$ sparsity penalty to identify unknown drift dynamics while maintaining closed-loop tracking.

The simulation results show that SP-VDF achieves lower terminal parameter estimation errors than SP-ICL and SP-VUF for the tested sparsity penalties. When the excitation decreases after the initial interval, directional forgetting maintains a full-rank MR and allows the parameter estimates to adapt after the parameter switch. SP-VDF also recovers the active parameter set, while the baseline methods exhibit estimator windup or numerical rank deficiency. The closed-loop tracking error remains small as the sparsity weight increases.

A limitation of the current formulation is the use of a constant sparsity weight under subspace excitation. When excitation is limited over extended intervals, the $\ell_1$ penalty can shrink parameters associated with weakly excited directions. Larger values of $\lambda$ also increase the computational stiffness of the online parameter update.

Future work will extend the SP-VDF framework to richer function representations, including Kolmogorov--Arnold networks, and investigate adaptive selection of the sparsity weight based on online measures of excitation. Robustness to unmodeled measurement noise and convergence for time-varying nonlinear systems will also be investigated.

\small
\bibliographystyle{IEEETrans.bst}
\bibliography{scc, sccmaster, scctemp, mybib}

\end{document}

%% file: figures/parameter_errors.tex
\begin{tikzpicture}

\begin{axis}[
    width=0.95\linewidth,
    height=0.650\linewidth,
    xlabel={$t\;(\mathrm{s})$},
    ylabel={$\| \tilde{\theta}(t) \|$},
    % ymode=log,
    grid=both,
    % TAC/SIAM typically uses dotted or very faint grids
    major grid style={draw=gray!40, dotted, line width=0.5pt},
    minor grid style={draw=gray!20, dotted, line width=0.5pt},
    tick style={black},
    axis line style={black},
    clip mode=individual,
    label style={font=\small},
    tick label style={font=\small},
    legend style={
    font=\footnotesize,
        draw=none,
        fill=none,
        at={(0.98,0.98)},       % Moves the legend to the top right (inside the plot)
        anchor=north east,      % Anchors the top right corner of the legend box
        legend columns=1,       % Changes back to a standard vertical list
        row sep=2pt,
    },
    % Enforce a single, consistent, thin line width for all plots
    every axis plot/.append style={line width=1.0pt}, 
]

% lambda = 0
\addplot[
    color=blue,
    solid,
]
table [x index=0, y index=1]
{data/lambda_0/data/parameter_estimation_error_norm.dat};
\addlegendentry{$\lambda=0$}

% lambda = 1e-3
\addplot[
    color=red,
    dashed,
]
table [x index=0, y index=1]
{data/lambda_0p001/data/parameter_estimation_error_norm.dat};
\addlegendentry{$\lambda=10^{-3}$}

% lambda = 5e-3
\addplot[
    color=green!60!black, % Darkened slightly so it's visible, but remains a basic green
    dashdotted,
    mark=*,
    mark size=1.2pt,
    mark repeat=14,
    mark options={solid}, % Keeps the marker outline solid even if line is dashed
]
table [x index=0, y index=1]
{data/lambda_0p005/data/parameter_estimation_error_norm.dat};
\addlegendentry{$\lambda=5\times10^{-3}$}

% lambda = 1e-2
\addplot[
    color=magenta,
    densely dotted,
]
table [x index=0, y index=1]
{data/lambda_0p01/data/parameter_estimation_error_norm.dat};
\addlegendentry{$\lambda=10^{-2}$}

% lambda = 5e-2
\addplot[
    color=orange,
    solid,
    mark=square,
    mark size=1.2pt,
    mark repeat=14,
]
table [x index=0, y index=1]
{data/lambda_0p05/data/parameter_estimation_error_norm.dat};
\addlegendentry{$\lambda=5\times10^{-2}$}

\end{axis}
\end{tikzpicture}

%% file: figures/tracking_errors.tex
\begin{tikzpicture}
\begin{axis}[
    width=0.95\linewidth,
    height=0.650\linewidth,
    xlabel={$t\;(\mathrm{s})$},
    ylabel={$\|e(t)\|$},
    % ymode=log,
    yminorticks=true,
    grid=major,
    % Replaced solid gray grid with traditional dotted grid
    major grid style={draw=gray!40, dotted, line width=0.5pt},
    tick style={black},
    axis line style={black},
    clip mode=individual,
    label style={font=\small},
    tick label style={font=\small},
    legend style={
    font=\footnotesize,
        draw=none,
        fill=none,
        at={(0.98,0.98)},       % Moves the legend to the top right (inside the plot)
        anchor=north east,      % Anchors the top right corner of the legend box
        legend columns=1,       % Changes back to a standard vertical list
        row sep=2pt,
    },
    % Enforce a single, consistent, thin line width for all plots
    every axis plot/.append style={line width=1.0pt}, 
]

% 1. lambda = 0
\addplot[
    color=blue, 
    solid
]
table [x index=0, y index=1] {data/lambda_0/data/tracking_error_norm.dat};
\addlegendentry{$\lambda=0$}

% 2. lambda = 1e-3
\addplot[
    color=red, 
    dashed
]
table [x index=0, y index=1] {data/lambda_0p001/data/tracking_error_norm.dat};
\addlegendentry{$\lambda=10^{-3}$}

% 3. lambda = 5e-3
\addplot[
    color=green!60!black, 
    dashdotted,
    mark=*,
    mark size=1.2pt,
    mark repeat=14,
    mark options={solid}
]
table [x index=0, y index=1] {data/lambda_0p005/data/tracking_error_norm.dat};
\addlegendentry{$\lambda=5\times 10^{-3}$}

% 4. lambda = 1e-2
\addplot[
    color=magenta, 
    densely dotted
]
table [x index=0, y index=1] {data/lambda_0p01/data/tracking_error_norm.dat};
\addlegendentry{$\lambda=10^{-2}$}

% 5. lambda = 5e-2
\addplot[
    color=orange, 
    solid,
    mark=square,
    mark size=1.2pt,
    mark repeat=14
]
table [x index=0, y index=1] {data/lambda_0p05/data/tracking_error_norm.dat};
\addlegendentry{$\lambda=5\times 10^{-2}$}

% % 6. lambda = 1e-1
% \addplot[
%     color=cyan, 
%     densely dashed, 
%     mark=triangle, 
%     mark size=1.4pt, 
%     mark repeat=12,
%     mark options={solid}
% ]
% table [x index=0, y index=1] {data/lambda_0p1/data/tracking_error_norm.dat};
% \addlegendentry{$\lambda=10^{-1}$}

\end{axis}
\end{tikzpicture}

%% file: figures/regressor_matrix_bounds.tex
\begin{tikzpicture}

\begin{axis}[
    width=0.95\linewidth,
    height=0.650\linewidth,
    xlabel={$t\;(\mathrm{s})$},
    ylabel={$\lambda_{\min}\left(\mathcal Y (t)\right)$},
    % ymode=log,
    grid=both,
    % TAC/SIAM traditional dotted grid
    major grid style={draw=gray!40, dotted, line width=0.5pt},
    minor grid style={draw=gray!20, dotted, line width=0.5pt},
    tick style={black},
    axis line style={black},
    clip mode=individual,
    label style={font=\small},
    tick label style={font=\small},
  legend style={
    font=\footnotesize,
    draw=none,
    fill=none,
    at={(0.98,0.98)},       % Moves the legend to the top right (inside the plot)
    anchor=north east,      % Anchors the top right corner of the legend box
    legend columns=1,       % Changes back to a standard vertical list
    row sep=2pt,
},
    % Enforce a single, consistent, thin line width for all plots
    every axis plot/.append style={line width=1.0pt}, 
]

% lambda = 0
\addplot[
    color=blue,
    solid,
]
table [x index=0, y index=1]
{data/lambda_0/data/information_matrix_eigenvalues.dat};
\addlegendentry{$\lambda=0$}

% lambda = 1e-3
\addplot[
    color=red,
    dashed,
]
table [x index=0, y index=1]
{data/lambda_0p001/data/information_matrix_eigenvalues.dat};
\addlegendentry{$\lambda=10^{-3}$}

% lambda = 5e-3
\addplot[
    color=green!60!black,
    dashdotted,
    mark=*,
    mark size=1.2pt,
    mark repeat=14,
    mark options={solid}, % Keeps the marker outline solid
]
table [x index=0, y index=1]
{data/lambda_0p005/data/information_matrix_eigenvalues.dat};
\addlegendentry{$\lambda=5\times10^{-3}$}

% lambda = 1e-2
\addplot[
    color=magenta,
    densely dotted,
]
table [x index=0, y index=1]
{data/lambda_0p01/data/information_matrix_eigenvalues.dat};
\addlegendentry{$\lambda=10^{-2}$}

% lambda = 5e-2
\addplot[
    color=orange,
    solid,
    mark=square,
    mark size=1.2pt,
    mark repeat=14,
    mark options={solid}, % Keeps the marker outline solid
]
table [x index=0, y index=1]
{data/lambda_0p05/data/information_matrix_eigenvalues.dat};
\addlegendentry{$\lambda=5\times10^{-2}$}

\end{axis}
\end{tikzpicture}

%% file: figures/regressor_matrix_bounds2.tex
\begin{tikzpicture}

\begin{axis}[
    width=0.95\linewidth,
    height=0.650\linewidth,
    xlabel={$t\;(\mathrm{s})$},
    ylabel={$\lambda_{\max}\left(\mathcal Y (t)\right)$},
    % xmin=0, xmax=100,
    %ymode=log,
    grid=both,
    % TAC/SIAM traditional dotted grid
    major grid style={draw=gray!40, dotted, line width=0.5pt},
    minor grid style={draw=gray!20, dotted, line width=0.5pt},
    tick style={black},
    axis line style={black},
    clip mode=individual,
    label style={font=\small},
    tick label style={font=\small},
    legend style={
    font=\footnotesize,
    draw=none,
    fill=none,
    at={(0.98,0.98)},       % Moves the legend to the top right (inside the plot)
    anchor=north east,      % Anchors the top right corner of the legend box
    legend columns=1,       % Changes back to a standard vertical list
    row sep=2pt,
},
    % Enforce a single, consistent, thin line width for all plots
    every axis plot/.append style={line width=1.0pt}, 
]

% lambda = 0
\addplot[
    color=blue,
    solid,
]
table [x index=0, y index=2]
{data/lambda_0/data/information_matrix_eigenvalues.dat};
\addlegendentry{$\lambda=0$}

% lambda = 1e-3
\addplot[
    color=red,
    dashed,
]
table [x index=0, y index=2]
{data/lambda_0p001/data/information_matrix_eigenvalues.dat};
\addlegendentry{$\lambda=10^{-3}$}

% lambda = 5e-3
\addplot[
    color=green!60!black,
    dashdotted,
    mark=*,
    mark size=1.2pt,
    mark repeat=14,
    mark options={solid},
]
table [x index=0, y index=2]
{data/lambda_0p005/data/information_matrix_eigenvalues.dat};
\addlegendentry{$\lambda=5\times10^{-3}$}

% lambda = 1e-2
\addplot[
    color=magenta,
    densely dotted,
]
table [x index=0, y index=2]
{data/lambda_0p01/data/information_matrix_eigenvalues.dat};
\addlegendentry{$\lambda=10^{-2}$}

% lambda = 5e-2
\addplot[
    color=orange,
    solid,
    mark=square,
    mark size=1.2pt,
    mark repeat=14,
    mark options={solid},
]
table [x index=0, y index=2]
{data/lambda_0p05/data/information_matrix_eigenvalues.dat};
\addlegendentry{$\lambda=5\times10^{-2}$}

\end{axis}
\end{tikzpicture}

%% file: figures/comparison_parameter_errors.tex
% --- Subplot A: Moderate Penalty ---
\begin{minipage}[t]{0.46\linewidth}
\vspace{0pt} % <-- MAGIC FIX: Forces absolute top alignment
\begin{tikzpicture}
\begin{axis}[
    width=\linewidth,
    height=0.75\linewidth,
    xlabel={$t\;(\mathrm{s})$},
    ylabel={$\|\tilde{\theta}(t)\|$},
    xmin=0, xmax=500,
    title={(a) Moderate Penalty ($\lambda = 10^{-2}$)},
    yminorticks=true,
    grid=major,
    major grid style={draw=gray!40, dotted, line width=0.5pt},
    tick style={black},
    axis line style={black},
    clip mode=individual,
    label style={font=\small},
    tick label style={font=\small},
    title style={font=\small, yshift=-1ex},
    every axis plot/.append style={line width=1.0pt},
    % Shared Legend - Centered below both plots and lowered to clear the xlabel
    legend style={
    font=\footnotesize,
    draw=none,
    fill=none,
    at={(0.98,0.98)},       % Moves the legend to the top right (inside the plot)
    anchor=north east,      % Anchors the top right corner of the legend box
    legend columns=1,       % Changes back to a standard vertical list
    row sep=2pt,
},
]

% 1. SP-VDF
\addplot[color=orange, solid]
table [x index=0, y index=1]
{data/lambda_0p01/data/parameter_estimation_error_norm.dat};
\addlegendentry{SP-VDF}

% 2. SP-VUF
\addplot[color=blue, dashed]
table [x index=0, y index=1]
{data_uniform/lambda_0p01/data/parameter_estimation_error_norm.dat};
\addlegendentry{SP-VUF}

% 3. SP-ICL
\addplot[color=green!60!black, densely dotted]
table [x index=0, y index=1]
{data_icl/lambda_0p01/data/parameter_estimation_error_norm.dat};
\addlegendentry{SP-ICL}

\end{axis}
\end{tikzpicture}
\end{minipage}\hfill% <-- Crucial percent sign prevents invisible spacing
% --- Subplot B: High Penalty ---
\begin{minipage}[t]{0.46\linewidth}
\vspace{0pt} % <-- MAGIC FIX: Forces absolute top alignment
\begin{tikzpicture}
\begin{axis}[
    width=\linewidth,
    height=0.75\linewidth,
    xlabel={$t\;(\mathrm{s})$},
    xmin=0, xmax=500,
    title={(b) High Penalty ($\lambda = 5 \times 10^{-2}$)},
    yminorticks=true,
    grid=major,
    major grid style={draw=gray!40, dotted, line width=0.5pt},
    tick style={black},
    axis line style={black},
    clip mode=individual,
    label style={font=\small},
    tick label style={font=\small},
    title style={font=\small, yshift=-1ex},
    every axis plot/.append style={line width=1.0pt},
    legend style={
    font=\footnotesize,
    draw=none,
    fill=none,
    at={(0.98,0.98)},       % Moves the legend to the top right (inside the plot)
    anchor=north east,      % Anchors the top right corner of the legend box
    legend columns=1,       % Changes back to a standard vertical list
    row sep=2pt,
},
]

% 1. SP-VDF
\addplot[color=orange, solid]
table [x index=0, y index=1]
{data/lambda_0p05/data/parameter_estimation_error_norm.dat};
\addlegendentry{SP-VDF}

% 2. SP-VUF
\addplot[color=blue, dashed]
table [x index=0, y index=1]
{data_uniform/lambda_0p05/data/parameter_estimation_error_norm.dat};
\addlegendentry{SP-VUF}

% 3. SP-ICL
\addplot[color=green!60!black, densely dotted]
table [x index=0, y index=1]
{data_icl/lambda_0p05/data/parameter_estimation_error_norm.dat};
\addlegendentry{SP-ICL}

\end{axis}
\end{tikzpicture}
\end{minipage}